\documentclass[10pt]{article}
\usepackage[lmargin=1in,rmargin=1in,tmargin=1in,bmargin=1in]{geometry}
\usepackage{amsmath,amsthm,amsfonts,amssymb,comment,slashed,mathtools}
\usepackage[english]{babel}
\usepackage[T1]{fontenc}  
\usepackage[utf8]{inputenc}
\usepackage{csquotes}
\usepackage{bm}
\usepackage[shortlabels]{enumitem}
\usepackage[affil-it]{authblk}  

\usepackage[section]{placeins}
\usepackage{soul} 

\usepackage{bm}
\usepackage{mathrsfs} 
\usepackage{graphicx}
\usepackage{stmaryrd}
\usepackage{epigraph}

\usepackage{hyperref}
\hypersetup{
colorlinks=true,
linkcolor=black,
urlcolor=black,
citecolor=black,
plainpages=false,
final}

\usepackage[backend=biber,style=alphabetic,giveninits=true,doi=false,isbn=false,url=false,sorting=nyt,maxbibnames=99,date=year]{biblatex}

\renewbibmacro{in:}{}

\usepackage{chngcntr}

\usepackage[capitalise,nameinlink]{cleveref}

\numberwithin{equation}{section}

\theoremstyle{plain}
\newtheorem{thm}{Theorem}
\newtheorem*{thmA}{Theorem A}
\newtheorem*{thmB}{Theorem B}
\newtheorem{prop}{Proposition}[section]
\newtheorem{lem}{Lemma}[section]
\newtheorem*{thm*}{Theorem}
\newtheorem{cor}{Corollary}
\newtheorem*{conj*}{Conjecture}
\newtheorem{conj}{Conjecture}
\newtheorem{question}{Question}
\theoremstyle{definition}
\newtheorem{defn}{Definition}[section]
\theoremstyle{remark}
\newtheorem{rk}{Remark}[section]
\newcounter{parentnumber}


\renewcommand{\Bbb}{\mathbb}

\newcommand{\ve}{\varepsilon}
\newcommand{\les}{\lesssim}

\newcommand{\diag}{\operatorname{diag}}

\newcommand{\tr}{\operatorname{tr}}

\newcommand{\Ric}{\mathrm{Ric}}

\newcommand{\Div}{\operatorname{div}}

\DeclareMathOperator{\adiv}{\slashed\Div}
\DeclareMathOperator{\ohat}{\hat\otimes}
\DeclareMathOperator{\arot}{\slashed{\mathrm{rot}}}

\begin{document}

\title{Event horizon gluing and black hole formation in vacuum: \\ the very slowly rotating case}

\author[1]{Christoph~Kehle\thanks{christoph.kehle@eth-its.ethz.ch}}
\author[2]{Ryan Unger\thanks{runger@math.princeton.edu}}
\affil[1]{\small  Institute~for~Theoretical~Studies \& Department of Mathematics, ETH~Zürich, 

8092~Zürich,~Switzerland \vskip.1pc \ 
}
\affil[2]{\small  Department of Mathematics, Princeton University, 
	Washington~Road,~Princeton~NJ~08544,~United~States~of~America \vskip.1pc \  
	}

 \date{\today}
\maketitle

\begin{abstract} In this paper, we initiate the study of characteristic event horizon gluing \emph{in vacuum}. More precisely, we prove that Minkowski space can be glued along a null hypersurface to \emph{any} round symmetry sphere in a Schwarzschild black hole  spacetime as a $C^2$ solution of the Einstein vacuum equations. The method of proof is fundamentally nonperturbative and is closely related to our previous work in spherical symmetry \cite{KU22} and Christodoulou's \emph{short pulse} method \cite{Christo09}. We also make essential use of the perturbative characteristic gluing results of Aretakis--Czimek--Rodnianski \cite{ACR1, Czimek2022-cl}.

 As an immediate corollary of our methods, we obtain characteristic gluing of Minkowski space to the event horizon of very slowly rotating Kerr with prescribed mass $M$ and specific angular momentum $a$.  Using our characteristic gluing results, we construct examples of vacuum gravitational collapse to very slowly rotating Kerr black holes in finite advanced time with prescribed $M$ and $0\le |a|\ll M$. 

 Our construction also yields the first example of a spacelike singularity arising from one-ended, asymptotically flat gravitational collapse in vacuum. 
\end{abstract}

\tableofcontents

\section{Introduction} \label{sec:introduction}

\emph{Characteristic gluing} is a powerful new method for constructing solutions of the \emph{Einstein field equations}
\begin{equation} \label{eq:Einstein-equations}
    \Ric(g) - \tfrac 12   R(g) g =   2 T
\end{equation} for a spacetime metric $g$ and coupled matter fields 
by gluing together two existing solutions along a null hypersurface. The setup of characteristic gluing is depicted in \cref{fig:char-gluing-setup} below and we will repeatedly refer to this diagram for definiteness. 

\begin{figure}[ht]
\centering{
\def\svgwidth{10pc}
\begingroup%
  \makeatletter%
  \providecommand\color[2][]{%
    \errmessage{(Inkscape) Color is used for the text in Inkscape, but the package 'color.sty' is not loaded}%
    \renewcommand\color[2][]{}%
  }%
  \providecommand\transparent[1]{%
    \errmessage{(Inkscape) Transparency is used (non-zero) for the text in Inkscape, but the package 'transparent.sty' is not loaded}%
    \renewcommand\transparent[1]{}%
  }%
  \providecommand\rotatebox[2]{#2}%
  \newcommand*\fsize{\dimexpr\f@size pt\relax}%
  \newcommand*\lineheight[1]{\fontsize{\fsize}{#1\fsize}\selectfont}%
  \ifx\svgwidth\undefined%
    \setlength{\unitlength}{161.65088563bp}%
    \ifx\svgscale\undefined%
      \relax%
    \else%
      \setlength{\unitlength}{\unitlength * \real{\svgscale}}%
    \fi%
  \else%
    \setlength{\unitlength}{\svgwidth}%
  \fi%
  \global\let\svgwidth\undefined%
  \global\let\svgscale\undefined%
  \makeatother%
  \begin{picture}(1,0.99999993)%
    \lineheight{1}%
    \setlength\tabcolsep{0pt}%
    \put(0,0){\includegraphics[width=\unitlength,page=1]{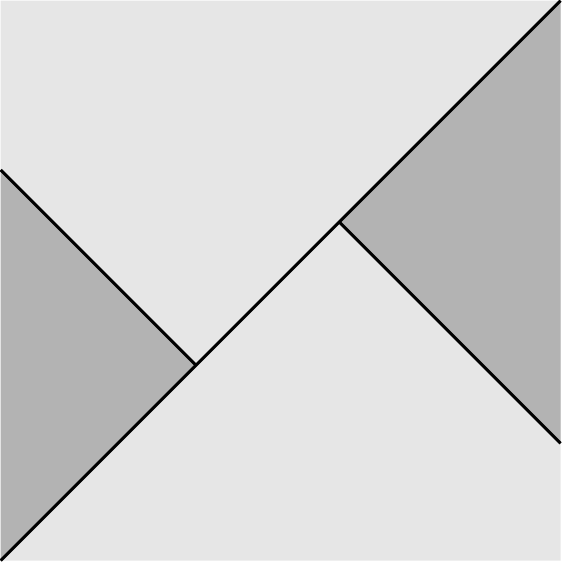}}%
    \put(0.43736499,0.48376128){\color[rgb]{0,0,0}\rotatebox{45}{\makebox(0,0)[lt]{\lineheight{1.25}\smash{\begin{tabular}[t]{l}$C$\end{tabular}}}}}%
    \put(0,0){\includegraphics[width=\unitlength,page=2]{char-gluing.pdf}}%
    \put(0.33621429,0.25091349){\color[rgb]{0,0,0}\makebox(0,0)[lt]{\lineheight{1.25}\smash{\begin{tabular}[t]{l}$S_1$\end{tabular}}}}%
    \put(0.57242169,0.67472436){\color[rgb]{0,0,0}\makebox(0,0)[lt]{\lineheight{1.25}\smash{\begin{tabular}[t]{l}$S_2$\end{tabular}}}}%
    \put(0.80856072,0.58474549){\color[rgb]{0,0,0}\makebox(0,0)[lt]{\lineheight{1.25}\smash{\begin{tabular}[t]{l}$\mathfrak R_2$\end{tabular}}}}%
    \put(0.10590686,0.33183049){\color[rgb]{0,0,0}\makebox(0,0)[lt]{\lineheight{1.25}\smash{\begin{tabular}[t]{l}$\mathfrak R_1$\end{tabular}}}}%
    \put(0.28938771,0.82717113){\color[rgb]{0,0,0}\makebox(0,0)[lt]{\lineheight{1.25}\smash{\begin{tabular}[t]{l}$(\mathcal M,g,\dotsc)$\end{tabular}}}}%
  \end{picture}%
\endgroup%
}
\caption{Penrose diagram depicting the setup of characteristic gluing. The null hypersurface $C$ is declared to be ``outgoing.''}
\label{fig:char-gluing-setup}
\end{figure}

In \cref{fig:char-gluing-setup}, the two dark gray regions $\mathfrak R_1$ and $\mathfrak R_2$ carry Lorentzian metrics and matter fields which satisfy  \eqref{eq:Einstein-equations}. The goal is to embed these regions into a spacetime $(\mathcal M,g,\dotsc)$ which satisfies \eqref{eq:Einstein-equations} globally, in the configuration depicted in \cref{fig:char-gluing-setup}. The characteristic gluing problem reduces to constructing characteristic data along a null hypersurface $C$ going between spheres $S_1\subset\mathfrak R_1$ and $S_2\subset\mathfrak R_2$, so that after constructing the light gray regions in \cref{fig:char-gluing-setup} by solving a characteristic initial value problem, the resulting spacetime is of the desired global regularity. 

Characteristic gluing is a useful tool for constructing spacetimes that share features of two existing solutions, and therefore display interesting behavior. This technique was recently pioneered for the \emph{Einstein vacuum equations}
\begin{equation}
\Ric(g)=0\label{eq:EVE}
\end{equation}
by Aretakis, Czimek, and Rodnianski \cite{ACR1}, and we will give an overview of their work in \cref{sec:ACR} below, including the obstruction-free gluing of Czimek--Rodnianski \cite{Czimek2022-cl}, which we will make crucial use of. Characteristic gluing was also recently used by the present authors to disprove the so-called \emph{third law of black hole thermodynamics} in the Einstein--Maxwell-charged scalar field model in spherical symmetry \cite{KU22}; see already \cref{sec:third-law}. We refer the reader to \cite{KU22} for a general formalism for the characteristic gluing problem. The present work is the first in a series of papers aimed at extending \cite{KU22} to the Einstein vacuum equations \eqref{eq:EVE}.

The most basic question in the study of characteristic gluing is the following:
\begin{question}\label{question-1}
    Which spheres $S_1$ and $S_2$ in which vacuum spacetimes can be characteristically glued as in \cref{fig:char-gluing-setup} as a solution of the Einstein vacuum equations \eqref{eq:EVE}? Are there any nontrivial obstructions? If so, can they be characterized geometrically? 
\end{question}

For example, a genuine obstruction arises from \emph{Raychaudhuri's equation} (see already \eqref{eq:Ray-v}), which implies that $S_2$ cannot be strictly outer untrapped if $S_1$ is (marginally) outer trapped. Another genuine obstruction arises from the rigidity of the positive mass theorem, which implies that if $\mathfrak R_2$ is Minkowski space, then $\mathfrak R_1$ is either Minkowski space or must be singular or incomplete in some sense. 

Our first theorem shows that the characteristic gluing of Minkowski space to (positive mass) Schwarzschild solutions is completely unobstructed, provided that we aim to glue to a symmetry sphere in Schwarzschild. In the statements of our theorems, refer to \cref{fig:char-gluing-setup}.

\begin{thm}\label{thm:Schwarzschild-gluing}
    Let $M>0$. Let $S_2$ be any non-antitrapped symmetry sphere in the Schwarzschild solution of mass $M$. Then $S_2$ can be characteristically glued to a sphere $S_1$ as depicted in \cref{fig:char-gluing-setup}, to order $C^2$ as a solution of the Einstein vacuum equations \eqref{eq:EVE}, where $S_1$ is a spacelike sphere in Minkowski space which is arbitrarily close to a round symmetry sphere. 
\end{thm}

For the precise statement of this theorem, see already \hyperlink{thmA}{Theorem A} in \cref{sec:thmA} below. Our method also immediately generalizes to very slowly rotating Kerr, and gives the following particularly clean statement about event horizon gluing:

\begin{thm}\label{thm:Kerr-gluing} There exists a constant $0<\mathfrak a_0\ll 1$  such that if $S_2$ is  a spacelike section of the event horizon of a Kerr black hole with mass $M>0$ and specific angular momentum $a$ satisfying $0\le|a|/M\le \mathfrak a_0$, then $S_2$ can be characteristically glued to a sphere $S_1$ as depicted in \cref{fig:char-gluing-setup}, to order $C^2$ as a solution of the Einstein vacuum equations \eqref{eq:EVE}, where $S_1$ is a spacelike sphere in Minkowski space which is close to a round symmetry sphere.
\end{thm}

This theorem is a special case of \hyperlink{thmB}{Theorem B} in \cref{sec:thmB} below.

\begin{rk}
    Strictly speaking, \hyperlink{thmB}{Theorem B} applies to Kerr coordinate spheres on the event horizon. However, it is easy to see that \emph{any} spacelike section on the event horizon can be connected to a Kerr coordinate sphere. 
\end{rk}

\begin{rk}
    There is an apparent asymmetry in the statements of \cref{thm:Schwarzschild-gluing} and \cref{thm:Kerr-gluing} about the allowable $S_2$'s. In fact, in \hyperlink{thmB}{Theorem B} below, we show that \emph{any} Kerr coordinate sphere can be connected to a sphere in Schwarzschild with smaller mass, but the maximum value of allowed angular momentum depends on the sphere in a non-explicit way that we prefer to explain later in the paper, see already  \cref{sec:ref-sphere-data}.
\end{rk}

\begin{rk}
    In \cref{thm:Schwarzschild-gluing}, the bottom sphere $S_1$ can be made arbitrarily close to an exact symmetry sphere in Minkowski, whereas in \cref{thm:Kerr-gluing}, the closeness to an exact symmetry sphere is limited by the size of $a/M$.
\end{rk}

\begin{rk}
The $C^2$ regularity of the spacetime metric in \cref{thm:Schwarzschild-gluing} and \cref{thm:Kerr-gluing} is due to limited regularity in the direction transverse to $C$. The regularity of the metric in directions tangent to $C$ can be made arbitrarily high (but finite). 
\end{rk}

By using \cref{thm:Kerr-gluing} and solving the Einstein equations backwards, we can construct examples of gravitational collapse to a black hole of prescribed very small angular momentum. The proof will be given in \cref{sec:proofs-of-corollaries} below. 

\begin{cor}[Gravitational collapse with prescribed $M$ and $0\le |a|\ll M$]\label{cor:main}
   Let $\mathfrak a_0$ be as in \cref{thm:Kerr-gluing}. Then for any mass $M>0$ and specific angular momentum $a$ satisfying $0\le|a|/ M\le\mathfrak a_0$, there exist one-ended asymptotically flat Cauchy data $(g_0,k_0)\in H^{7/2-}_\mathrm{loc}\times H^{5/2-}_\mathrm{loc}$ for the Einstein vacuum equations \eqref{eq:EVE} on $\Sigma\cong\Bbb R^3$, satisfying the constraint equations, such that the maximal future globally hyperbolic development $(\mathcal M^4,g)$ contains a black hole $\mathcal{BH}\doteq \mathcal M\setminus J^-(\mathcal I^+)$ and has the following properties: 
 \begin{itemize}
     \item The Cauchy surface $\Sigma$ lies in the causal past of future null infinity, $\Sigma\subset J^-(\mathcal I^+)$. In particular, $\Sigma$ does not intersect the event horizon $\mathcal H^+\doteq\partial(\mathcal{BH})$ or contain trapped surfaces. 
      \item $(\mathcal M,g)$ contains trapped surfaces. 
     \item For sufficiently late advanced times $v\ge v_0$, the domain of outer communication, including the event horizon $\mathcal H^+$, is isometric to that of a Kerr solution with parameters $M$ and $a$. For $v\ge v_0$, the event horizon of the spacetime can be identified with the event horizon of Kerr. 
 \end{itemize}
\end{cor}

 \begin{figure}[ht]
\centering{
\def\svgwidth{20pc}
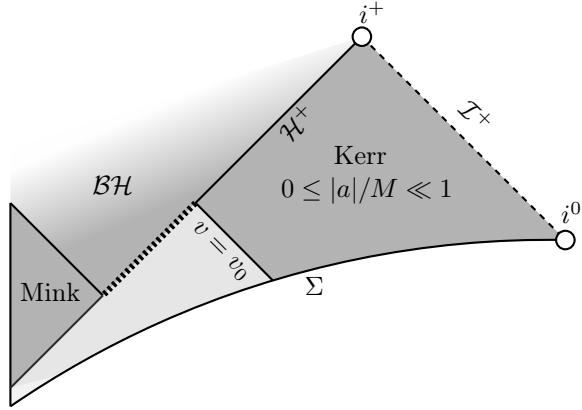}
\caption{Penrose diagram for \cref{cor:main}. The textured line segment is where the gluing data constructed in \cref{thm:Kerr-gluing} live.}
\label{main-corollary-diagram}
\end{figure}

For the relevant Penrose diagram, consult \cref{main-corollary-diagram} below. 

\begin{rk}
    It is a classical result that the Einstein equations are well posed in $H^{7/2-}_\mathrm{loc}\times H^{5/2-}_\mathrm{loc}$, see \cite{HKM} and also \cite{Planchon-Rodnianski, C-low-reg}. 
\end{rk}

By performing characteristic gluing as in \cref{thm:Schwarzschild-gluing} of Minkowski space to spheres in a Schwarzschild solution lying just inside the horizon and using Cauchy stability, we also obtain:

\begin{cor}[Gravitational collapse with a spacelike singularity]\label{cor:spacelike}
      There exist one-ended asymptotically flat Cauchy data $(g_0,k_0)\in H^{7/2-}_\mathrm{loc}\times H^{5/2-}_\mathrm{loc}$ for the Einstein vacuum equations \eqref{eq:EVE} on $\Sigma\cong\Bbb R^3$, satisfying the constraint equations, such that the maximal future globally hyperbolic development $(\mathcal M^4,g)$ contains a black hole $\mathcal{BH}\doteq \mathcal M\setminus J^-(\mathcal I^+)$ and has the following properties: 
 \begin{itemize}
     \item The Cauchy surface $\Sigma$ lies in the causal past of future null infinity, $\Sigma\subset J^-(\mathcal I^+)$. In particular, $\Sigma$ does not intersect the event horizon $\mathcal H^+\doteq\partial(\mathcal{BH})$ or contain trapped surfaces. 
     \item For sufficiently late advanced times $v\ge v_0$, the domain of outer communication, together with a full double null slab lying in the interior of the black hole, is isometric to a portion of a Schwarzschild solution as depicted in \cref{fig:spacelike-singularity}. The double null slab terminates in the future at a spacelike singularity, isometric to the ``$r=0$'' singularity in Schwarzschild. 
 \end{itemize}
\end{cor}

We will sketch the proof of this result in \cref{sec:proofs-of-corollaries} below. 

 \begin{figure}[ht]
\centering{
\def\svgwidth{18pc}
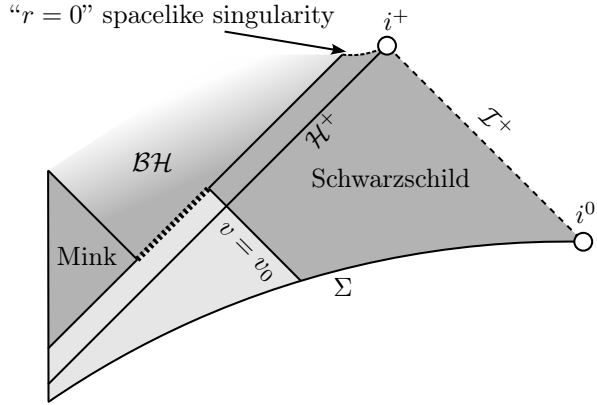}
\caption{Example of gravitational collapse with a piece of a spacelike singularity emanating from timelike infinity $i^+$. Characteristic gluing is performed along the textured line segment, where the top gluing sphere has radius very close to the Schwarzschild radius of the black hole to be formed.}
\label{fig:spacelike-singularity}
\end{figure}

\subsection{The characteristic gluing problem}\label{sec:previous-work}

The study of the characteristic gluing problem was initiated by Aretakis for the linear scalar wave equation
\begin{equation}
    \Box_g\phi=0 \label{eq:linear-wave}
\end{equation}
on general spacetimes $(\mathcal M^{3+1},g)$ in \cite{AretakisLinear}. Aretakis showed that there is always a finite-dimensional (but possibly trivial) space of obstructions to the characteristic gluing problem. More precisely, he showed that there are at most finitely many (possibly none) \emph{conserved charges} that are computed from the given solutions at $S_1$ and $S_2$ in \cref{fig:char-gluing-setup} that determine whether characteristic gluing can be performed. These charges are conserved along $C$ for any solution of \eqref{eq:linear-wave}. This gives a definitive answer\footnote{\cite{AretakisLinear} only deals with $C^1$ characteristic gluing, but is definitive in this regularity class.} to \cref{question-1} for \eqref{eq:linear-wave}: There is a precise characterization of which spheres can be glued---the matching of all conserved charges is both necessary and sufficient. 

\subsubsection{Characteristic gluing for the Einstein vacuum equations}\label{sec:ACR}

Characteristic gluing for the Einstein vacuum equations \eqref{eq:EVE} was initiated by Aretakis, Czimek, and Rodnianski in a fundamental series of papers \cite{ACR1, ACR2, ACR3}. They study the perturbative regime around Minkowski space, that is, when both spheres $S_1$ and $S_2$ in \cref{fig:char-gluing-setup} are close to symmetry spheres in Minkowski space. The strategy employed is to linearize the Einstein equations around Minkowski space in the framework of Dafermos--Holzegel--Rodnianski \cite{DHR19}, solve the characteristic gluing problem for the linearized Einstein equations, and then conclude a small-data nonlinear gluing result by an implicit function theorem argument. 

In the course of their argument, they discover that the linearized Einstein equations around Minkowski space in double null gauge possess \emph{infinitely many conserved charges}. However, it turns out that all but ten of these charges are due to \emph{gauge invariance} of the Einstein equations (cf.~the \emph{pure gauge solutions} of \cite{DHR19}). The remaining charges, which we define precisely in \cref{def:charges} below, are genuine obstructions to the linear characteristic gluing problem, and must therefore be assumed to be equal at $S_1$ and $S_2$ in order for the inverse function scheme to yield a genuine solution. 

The conserved charges of Aretakis--Czimek--Rodnianski can be identified with the ADM energy, linear momentum, angular momentum, and center of mass. This identification is used in \cite{ACR3} to give a new proof of the spacelike gluing results of \cite{Corvino, CorvinoSchoen, Carlotto} using characteristic gluing. 

Later, Czimek and Rodnianski \cite{Czimek2022-cl} made the fundamental observation that the linear conservation laws can be violated at the nonlinear level by certain explicit ``high frequency'' seed data for the characteristic initial value problem.\footnote{The high frequency perturbations are becoming singular in the Minkowski limit and hence do not linearize in a regular fashion.} They then use these high frequency perturbations to \emph{adjust} the linearly conserved charges in the full nonlinear theory, so that the main theorem of \cite{ACR1} applies. The result, which we state as \cref{thm:CR} in \cref{sec:CR} below, is that two spheres close to Minkowski space can be glued if the differences of the conserved charges satisfy a certain \emph{coercivity} condition. Roughly, the assumption is that the change in the Hawking mass be larger than the changes in the other conserved charges and that the change in angular momentum is itself much smaller than the distance of $S_1$ and $S_2$ to spheres in Minkowski space. Their result has the remarkable corollary of obstruction-free spacelike gluing of asymptotically flat Cauchy data to Kerr in the far region.

We note at this point that the analysis of \cite{ACR1,ACR2,ACR3,Czimek2022-cl} is limited to $C^2$ regularity in the ingoing direction $u$, but allows for arbitrarily high regularity in $v$ and angular directions.\footnote{In this paragraph we are referring to the results in \cite{ACR1,ACR2,ACR3,Czimek2022-cl} that have to do with characteristic gluing as is depicted in \cref{fig:char-gluing-setup}. Aretakis--Czimek--Rodnianski also consider another type of characteristic gluing, \emph{bifurcate characteristic gluing}, which works to arbitrarily high order of differentiability.} It is not clear whether their analysis (especially \cite{Czimek2022-cl}) can be generalized to higher order transverse derivatives. 

The linearized characteristic gluing problem for \eqref{eq:EVE} was redone in Bondi gauge and extended to incorporate a cosmological constant and different topologies of the cross sections of the null hypersurface $C$ by Chru\'sciel and Cong \cite{Chrusciel-Cong}. This work also addresses linearized characteristic gluing of higher order transverse derivatives. 

\begin{question}
    Is there a general geometric characterization of conservation laws associated to the linearized Einstein equations around a fixed background? Is there always a finite number of conservation laws? Is the generic spacetime free of conservation laws at the linear level? 
\end{question}

One might also wonder if there is a precise connection between the conservation laws observed in the null setting with the cokernel of the linearized constraint map studied in the spacelike gluing problem \cite{Corvino, CorvinoSchoen, ChruscielDelay}. We refer the reader to \cite{Gluing-variations} for more discussion about these issues. 

\subsubsection{Characteristic gluing in spherical symmetry}

The present authors have studied the characteristic gluing problem for the Einstein--Maxwell-charged scalar field system in spherical symmetry \cite{KU22}. Our main theorem can be stated as follows:

\begin{thm}[Theorem 2 in \cite{KU22}]
\label{thm-informal-statement}
Let $k\in \Bbb N$ be a regularity index, $\mathfrak q\in [-1,1]$  a charge to mass ratio,  and  $\mathfrak e\in \Bbb R\setminus\{0\}$ a fixed coupling constant. For any $M_f$ sufficiently large depending on $k$, $\mathfrak q$, and $\mathfrak e$, and any $0\le M_i\le \frac 18 M_f$, and $2M_i<R_i\le \tfrac 12 M_f$, 
there exist spherically symmetric characteristic data for the Einstein--Maxwell-charged scalar field system with coupling constant $\mathfrak e$ gluing the Schwarzschild symmetry sphere of mass $M_i$ and radius $R_i$ to the Reissner--Nordstr\"om event horizon with mass $M_f$ and charge $e_f=\mathfrak{q}M_f$ up to order $k$. 
\end{thm}

Our proof is fundamentally \emph{nonperturbative} in that we work directly with the nonlinear equations and do not require a perturbative analysis. This is made possible by three ingredients:
\begin{enumerate}[(i)]
    \item The Hawking mass is glued by judiciously initiating the transport equations at $S_1$ \emph{or} $S_2$ and directly exploiting gauge freedom in the form of boosting the double null gauge by hand.  

    \item The charge of the Maxwell field is glued by exploiting a monotonicity property of Maxwell's equation specific to spherical symmetry.

    \item Transverse derivatives of the scalar field are glued by exploiting a parity symmetry of the Einstein--Maxwell-charged scalar field system specific to spherical symmetry and invoking the Borsuk--Ulam theorem. 
\end{enumerate}

The argument in the present paper has two crucial ingredients: the implementation of  idea (i) above in the context of the Einstein vacuum equations, and the obstruction-free characteristic gluing of Czimek--Rodnianski \cite{Czimek2022-cl}, which replaces the Borsuk--Ulam argument in vacuum. See already \cref{sec:proof-outline} for the outline of our proof.

\subsection{The third law of black hole thermodynamics in vacuum}\label{sec:third-law}

The \emph{third ``law'' of black hole thermodynamics} is the conjecture that a subextremal black hole cannot become extremal in finite time by any continuous process, \underline{no matter how idealized}, in  which the spacetime and matter fields remain regular and obey the weak energy condition \cite{BCH, Israel-third-law}. The main application of our previous characteristic gluing result in \cite{KU22} is the following definitive \emph{disproof} of the third law: 

\begin{thm}[Theorem 1 in \cite{KU22}]\label{thm:third-law}
The ``third law of black hole thermodynamics'' is false. More precisely, subextremal black holes can become extremal in finite time, evolving from regular initial data. In fact, there exist regular one-ended Cauchy data for the Einstein--Maxwell-charged scalar field system 
which undergo gravitational collapse and form an exactly Schwarzschild apparent horizon, only for the spacetime to form an exactly extremal Reissner--Nordstr\"om event horizon at a later advanced time. 
\end{thm}

We refer the reader to \cite{KU22} for an extensive discussion of the history and physics of the third law. The black holes in \cref{thm:third-law} are constructed in two stages: First the scalar field is used to form an exact Schwarzschild apparent horizon in finite time, which is then charged up to extremality by exploiting the coupling of the scalar field with the electromagnetic field. \cref{thm:Schwarzschild-gluing} above can be viewed as a generalization of this first step. We conjecture that the second step can also be generalized to vacuum: 

\begin{conj}
The Schwarzschild symmetry sphere of mass $M_i$ and radius $R_i$ can be characteristically glued to any non-antitrapped Kerr coordinate sphere with radius $R_f\gg R_i$ in a Kerr solution with mass $M_f\gg M_i$ and specific angular momentum $0\le |a_f|\le M_f$.
\end{conj}

If this conjecture holds, Schwarzschild can be spun up to extremality. Arguing as in \cite{KU22} and \cref{cor:main} in the present paper, this would imply

\begin{conj}\label{conj:third-law-vac} The ``third law of black hole thermodynamics'' is already false in vacuum. More precisely, there exist regular one-ended Cauchy data for the  Einstein vacuum equations \eqref{eq:EVE} which undergo gravitational collapse and form an exactly Schwarzschild apparent horizon, only for the spacetime to form an exactly extremal Kerr event horizon at a later advanced time. 
\end{conj}

\begin{rk}
    By using negatively charged pulses in \cite{KU22}, we can design characteristic data that also ``discharges'' the black hole. It would be very interesting to find a mechanism that can both ``spin up'' and ``spin down'' a Kerr black hole, or move the rotation axis without changing the angular momentum much. 
\end{rk}

\begin{rk}
    It is not possible to have a solution of the pure Einstein--Maxwell equations which behaves like one of the solutions in \cref{thm:third-law}. This is because the vacuum Maxwell equation ${d}{\star}F=0$ always gives rise to a conserved electric charge $(4\pi)^{-1}\int_S \star F$, even outside of spherical symmetry. On Schwarzschild, this charge is zero, and on Reissner--Nordstr\"om, it equals the charge parameter $e$. 
\end{rk}

\begin{rk}
    Similarly, if a vacuum spacetime has an axial Killing field $Z$, then the Komar angular momentum $(16\pi)^{-1}\int_S \star dZ^\flat$ is conserved. Therefore, \cref{conj:third-law-vac} cannot be proved in axisymmetry. 
\end{rk}

\subsection{Outline of the proof}\label{sec:proof-outline}

In this section we give a very brief outline of the proof of \cref{thm:Schwarzschild-gluing}. The gluing is performed in two stages and should be thought of as being performed backwards in time. 

 \begin{figure}[ht]
\centering{
\def\svgwidth{13pc}
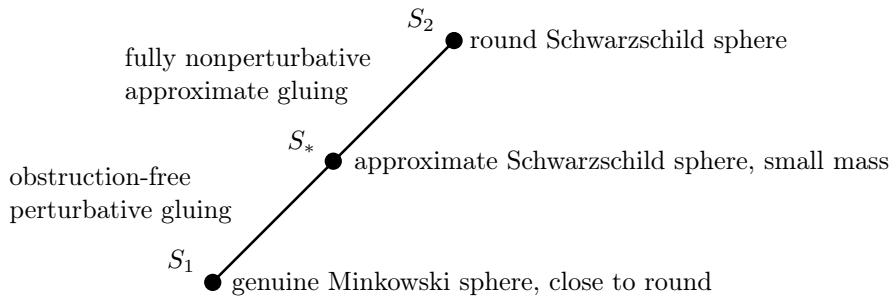}
\caption{Two-step process for the proof of \cref{thm:Schwarzschild-gluing}.}
\label{fig:two-phases}
\end{figure}

First, a fully nonperturbative mechanism is used to connect the exact Schwarzschild sphere $S_2$ of mass $M$ and radius $R$ to a sphere data set $S_*$ which is very close to a Schwarzschild sphere of mass $0<M_* \ll R$ and radius $\approx R$. See already \cref{prop:Schw-perturb-Schw} below. In the second stage of the gluing, we use the main theorem of \cite{Czimek2022-cl}, which we state below as \cref{thm:CR}, as a black box. In order to satisfy the necessary coercivity conditions required for obstruction-free characteristic gluing, we choose $0< \ve_\sharp\ll M_*\ll R$, where $\ve_\sharp$ measures the closeness of $S_*$ to the $(M_*,R)$-Schwarzschild sphere. 

The nonperturbative mechanism which glues $S_*$ to $S_2$ involves the injection of two pulses\footnote{By the Poincar\'e--Hopf theorem, the shear $\hat\chi$ vanishes somewhere on each sphere. This means that in any characteristic gluing problem in vacuum where the null expansion $\tr\chi$ has to change everywhere, the zero of $\hat\chi$ has to move along the sphere as $v$ increases. In our setting, we choose the first pulse to be supported away from the north pole and the second pulse to be supported away from the south pole.} of gravitational waves (described mathematically by the shear $\hat\chi$) along $C=[0,1]_v\times S^2$, of amplitude $\delta^{1/2}=O(\ve_\sharp)$, together with a choice of outgoing null expansion $\tr\chi$ at $S_2$ such that $|\!\tr\chi|\les \delta$. In order to fix the Hawking mass of $S_2$ to be $M$, the ingoing null expansion $\tr\underline\chi$ is then chosen to be $\approx -\delta^{-1}$ at $S_2$.\footnote{$\tr\underline\chi<0$ on $S_2$ means that $S_2$ is not antitrapped.} The Hawking mass of $S_*$ is fixed to be arbitrarily close to $M_*$ by tuning $\hat\chi$ and using the monotonicity of Raychaudhuri's equation; see already \cref{lem:construction-of-lambda}.\footnote{In \cite{KU22}, we use a similar monotonicity to glue the charge.}  We then step through the null structure equations and Bianchi identities as in \cite[Chapter 2]{Christo09} and establish a $\delta$-weighted hierarchy for the sphere data at $S_*$; see already \cref{lem:bottom-sphere-estimates}. Finally, we boost the cone by $\delta$, which brings $S_2$ to a reference Schwarzschild sphere and $S_*$ within $\ve_\sharp$ of a reference Schwarzschild sphere with mass $M_*$. This construction may be thought of as a direct adaptation, in vacuum, of the idea used to prove Schwarzschild event horizon gluing in spherical symmetry for the Einstein-scalar field system in \cite[Section 4.1]{KU22}.

\subsubsection{Relation to Christodoulou's short pulse method}\label{sec:short-pulse-intro}

After the boost, one can  interpret the above approximate nonperturbative gluing mechanism as a ``short pulse'' data set defined on $[0,\delta]\times S^2$ as in \cite{Christo09}, but fired backwards. In this context, our equation \eqref{eq:average} below should be compared with the condition (4.1) in \cite{LiYu} (see also \cite{LiMei, Athanasiou2020}). This condition guarantees that certain components of the sphere data at $S_*$ are \emph{a posteriori} closer to spherical symmetry.

Likewise, our \cref{cor:main} can be compared with the main theorem of Li and Mei in \cite{LiMei}. In particular, we also prove trapped surface formation starting from Cauchy data outside of the black hole region. Their proof utilizes the trapped surface formation mechanism of Christodoulou \cite{Christo09} and Corvino--Schoen spacelike gluing \cite{CorvinoSchoen}. 

Our proof of trapped surface formation starting from Cauchy data is fundamentally different from \cite{LiMei} because it does not appeal to Christodoulou's trapped surface formation mechanism \cite{Christo09}. In fact, the only aspect of the evolution problem we require is Cauchy stability. Furthermore, we can directly prescribe the (very slowly rotating) Kerr parameters of the black hole to be formed. In particular, we may take $a=0$, which guarantees the existence of a spacelike singularity, see \cref{cor:spacelike}. However, our data is of limited regularity (but still in a well-posed class). Nevertheless, by appealing to Cauchy stability once again, \cref{cor:main} has the further corollary of trapped surface formation starting from an \emph{open set} of Cauchy data. 

 \subsection*{Acknowledgments}

The authors thank Mihalis Dafermos and Igor Rodnianski for helpful discussions. C.K.~acknowledges support by Dr.~Max R\"ossler, the Walter Haefner Foundation, and the ETH Z\"urich Foundation.  R.U.~thanks the University of Cambridge and ETH Z\"urich for hospitality as this work was being carried out. 

\section{Characteristic initial data and characteristic gluing}  

In this section, we give a brief review of the characteristic gluing problem for the Einstein vacuum equations \eqref{eq:EVE} in double null gauge \cite{ACR1, ACR2, ACR3, Czimek2022-cl}.  We follow the conventions of \cite{Czimek2022-cl} unless otherwise stated.

If the reader is unfamiliar with double null gauge, we defer to \cite{Christo09} for exposition. We have collected formulas and notions used explicitly in the present paper in \cref{app:A}.

\subsection{Sphere data, null data, and seed data}

The terminology used in this paper is in agreement with \cite{Czimek2022-cl}, which we will be using as a black box, and therefore differs slightly from our previous paper \cite{KU22}. We hope this facilitates the reader in understanding exactly how the main notions and results from \cite{Czimek2022-cl} are being used here. 

\subsubsection{Sphere data} \label{sec:sphere-data}

Given a solution $(\mathcal M^4,g)$ of the Einstein vacuum equations \eqref{eq:EVE} and a sphere $S$ in a double null foliation, the 2-jet of $g$ can be determined from knowledge of the metric coefficients, Ricci coefficients, and curvature components. However, the equations themselves, such as the Codazzi equation \eqref{eq:Codazzi1} allow some of these degrees of freedom to be computed from the others, just in terms of derivatives tangent to $S$. This leads to the following definition:

\begin{defn}[$C^2$ sphere data, {\cite[Definition 2.4]{ACR2}}]\label{def:C2-data}
Let $S$ be a $2$-sphere. Sphere data $x$ on $S$ consists of a choice of round metric $\gamma$ on $S$ and the following tuple of $S$-tensors
\begin{equation}
    x= (\Omega, \slashed g, \Omega \tr\chi , \hat \chi, \Omega \tr \underline\chi, \hat{\underline \chi}, \eta, \omega, D \omega,\underline \omega, \underline{D \omega} , \alpha, \underline \alpha), \label{eq:sphere-defn}
\end{equation}
where $\Omega$ is a positive function, $\slashed g$ a Riemannian metric, $\Omega \tr\chi , \hat \chi, \Omega \tr \underline\chi, \omega, D \omega,\underline \omega, \underline{D \omega} $ are scalar functions, $\eta$ is a vector field, $\hat \chi, \hat{\underline \chi}$, $\alpha$ and $\underline \alpha$ are symmetric $\slashed g$-traceless $2$-tensors. 
\end{defn}

\begin{defn}[$C^2_u C^{2+m}_{v}$ sphere data, 
 {\cite[Definition 2.28]{ACR2}}]\label{def:higher-order-data}
Let $S$ be a $2$-sphere and $m\ge 0$ an integer. Higher order sphere data $x$ on $S$ consists of a choice of round metric $\gamma$ on $S$, the tuple \eqref{eq:sphere-defn}, together with
\begin{equation}
    (\hat D\alpha,\dotsc, \hat D^m\alpha,D^2\omega,\dotsc,D^{m+1}\omega),\label{eq:higher}
\end{equation}
where $\hat D^j\alpha$ are symmetric $\slashed g$-traceless $2$-tensors and $D^j\omega$ scalar functions. We will write $x=(x^\mathrm{low},x^\mathrm{high})$, where $x^\mathrm{low}$ is a $C^2$ sphere data set and $x^\mathrm{high}$ denotes the tuple \eqref{eq:higher}. 
\end{defn}

When sphere data is obtained from a geometric sphere in a vacuum spacetime, one has to make a choice of normal null vector fields $L$ and $\underline L$. See \cref{lem:data-generation} below, for instance. As is well known, the null pair $\{L,\underline L\}$ can be ``boosted''  by the transformation 
\begin{equation}
\tilde L \doteq \frac{1}{\lambda} L,\quad \tilde{\underline L}\doteq \lambda \underline L,\label{eq:boost-1}
\end{equation} where $\lambda\in\Bbb R_+$. This boost freedom was also quite useful in the preceding paper \cite{KU22}.

\begin{defn}[Boosted sphere data]\label{def:boost}
Let $x$ be a sphere data set as in \cref{def:higher-order-data} and $\lambda\in\Bbb R_+$. Then the $\lambda$-\emph{boosted sphere data set} is the $C^2_uC^{2+m}_v$ sphere data set given by 
\begin{align*}
    \mathfrak b_\lambda(x^\mathrm{low})&\doteq (\Omega, \slashed g, \lambda^{-1}\Omega \tr\chi , \lambda^{-1}\hat \chi, \lambda\Omega \tr \underline\chi,\lambda \hat{\underline \chi}, \eta, \lambda^{-1}\omega, \lambda^{-2}D \omega,\lambda\underline \omega,\lambda^2 \underline{D \omega} , \lambda^{-2}\alpha, \lambda^2\underline \alpha), \\
     \mathfrak b_\lambda(x^\mathrm{high})&\doteq (\lambda^{-3}\hat D\alpha,\dotsc,\lambda^{-2-m}\hat D^m\alpha,\lambda^{-3}D\omega,\dotsc,\lambda^{-2-m}D^{m+1}\omega).
\end{align*}
This is the effect that the boost \eqref{eq:boost-1} has on the metric coefficients, Ricci coefficients, and curvature components in double null gauge. 
\end{defn}

There is a norm $\|x\|_{\mathcal X^m}$ defined on $C^{2}_uC^{2+m}_{v}$ sphere data sets employed in \cite{ACR2, Czimek2022-cl}, which is just a sum of high order (in the angular variable $\theta$) Sobolev norms of the sphere data components \cite[Definition 2.5]{ACR2}. We will show very strong pointwise smallness for arbitrary numbers of angular derivatives later and thus will not need the exact form of these norms in order to apply the result of \cite{Czimek2022-cl}. 

\begin{defn}[Sphere diffeomorphisms]\label{def:sphere-diffeo}
    Given a diffeomorphism $\psi:S^2\to S^2$, we let $\psi$ act on $C^2_uC^{2+m}_v$ data sets by pullback on each component. 
\end{defn}

\subsubsection{Null data}

\begin{defn}[Ingoing and outgoing null data {\cite[Definition 2.6]{ACR2}}]
Let $v_1<v_2$. An \emph{outgoing null data set} is given by an assignment $v\mapsto x(v)$, where $x(v)$ is a $C^2$ sphere data set. We may say that the null data lives on the null hypersurface $C\doteq C^{[v_1,v_2]}\doteq [v_1,v_2]\times S^2$. An \emph{ingoing null data set} is defined in the same way, but with the formal variable $v$ replaced with $u$ and $\eta$ replaced by $\underline\eta$ in \eqref{eq:sphere-defn}. 

Higher order null data is defined in the obvious way, with $x(v)$ being $C^2_uC^{2+m}_{v}$ sphere data. Null data on its own is not assumed to satisfy the null structure equations and Bianchi identities. 
\end{defn}

There are several norms on null data that are employed in \cite{ACR1,ACR2,ACR3,Czimek2022-cl}. These include the standard norm $\mathcal X$ defined on ingoing and outgoing null data the high regularity norm $\mathcal X^+$ defined on ingoing null data, and the high frequency norm $\mathcal X^{\text{h.f.}}$ defined on outgoing null data using in obstruction-free characteristic gluing. As we will not need the precise forms of these norms in the present work, we refer the reader to \cite[Definition 2.7]{ACR2} for details. 

\subsubsection{Christodoulou seed data}\label{sec:Chr-seed-data}

We will employ the following method, originating in the work of Christodoulou \cite{Christo09}, for producing solutions of the null structure and Bianchi equations along a null hypersurface $C$.  

For definiteness, we seek a solution of the null constraints on the null cone segment $C\doteq C^{[0,1]}\doteq [0,1]\times S^2$. The coordinate along $[0,1]$ is called $v$ and we set $L=\partial_v$. On $S^2$, we have the round metric
\begin{equation*}
    \gamma \doteq d\vartheta^2+\sin^2\vartheta\,d\varphi^2,
\end{equation*}
where $(\vartheta,\varphi)$ are standard spherical polar coordinates. We interpret $\gamma$ as a symmetric $S$-tensor on $C$ (see \cref{app:A} for this terminology) by imposing $\gamma(\partial_v,\cdot)=0$. We set $S_v\doteq \{v\}\times S^2$. 

\begin{lem}\label{lem:seed-data} Let $\hat{\slashed g}$ be a smooth $S$-(0,2)-tensor field on $C$ which induces a Riemannian metric on the sections of $C$ satisfying 
\begin{equation}
    \tr_{\hat{\slashed g}}D\hat{\slashed g}=0,\label{eq:traceless}
\end{equation}
where $D\hat{\slashed g}\doteq \slashed{\mathcal L}{}_L\hat{\slashed g}$ as in \cref{app:A}.\footnote{Concretely, this means $\hat{\slashed g}=\hat{\slashed g}(v)$ is a smooth 1-parameter family of Riemannian metrics on $S^2$. We identify $S^2$ with $S_v\subset C$. Since $b=0$, $\slashed {\mathcal L}{}_L\hat{\slashed g}=\mathcal L_L\hat{\slashed g}$ and relative to any angular coordinates $\vartheta^A$ defined on $S_1$ extended to $C$ according to $L\vartheta^A=0$, $(D\hat{\slashed g})_{AB}=\partial_v(\hat{\slashed g}{}_{AB})$.}
Let $\slashed g{}_1$ be a Riemannian metric on $S_1$ which is conformal to $\hat{\slashed g}(1)$, $\tr\chi_1$ and $\tr\underline\chi_1$ be functions on $S_1$, $\eta_1$ be a 1-form on $S_1$, and $\hat{\underline\chi}_1$ and $\underline\alpha_1$ be ${\slashed g}{}_1$-traceless symmetric 2-tensors on $S_1$.  

Then there exist uniquely determined ${\slashed g}{}_1$-traceless symmetric 2-tensors $\hat\chi_1$ and $\alpha_1,\hat{D}\alpha_1\dotsc,\hat{D}^m\alpha_1$ on $S_1$, $v_0\in (-1,1)$, and null data 
\begin{align*}
   x^\mathrm{low}(v) &=(\Omega, \slashed g, \Omega \tr\chi , \hat \chi, \Omega \tr \underline\chi, \hat{\underline \chi}, \eta, \omega, D \omega,\underline \omega, \underline{D \omega} , \alpha, \underline \alpha),\\
   x^\mathrm{high}(v)&=(\hat D\alpha,\dotsc, \hat D^m\alpha,D^2\omega,\dotsc,D^{m+1}\omega),
\end{align*}
defined for $v\in (v_0,1]\cap [0,1]$, satisfying the null structure equations and Bianchi identities along $C$, such that 
\begin{align}
      \label{eq:x(1)}  x^\mathrm{low}(1)&= (1,\slashed g{}_1,\tr\chi_1,\hat\chi_1,\tr\underline\chi_1,\hat{\underline\chi}_1,\eta_1,0,0,0,0,\alpha_1,\underline\alpha{}_1)\\
        x^{\mathrm{high}}(1)&=(\hat D\alpha_1,\dotsc,\hat D^{m+1}\alpha_1,0,\dotsc,0), 
\end{align}
For every $v\in (v_0,1]\cap [0,1]$, $\slashed g(v)$ is conformal to $\hat{\slashed g}(v)$ and $\Omega^2(v)=1$ identically, so $D^j\omega(v)=0$ identically for $j=0,\dotsc,m+1$. The number $v_0$ is either strictly negative (in which case $x$ exists on all of $C$), or is nonnegative and satisfies 
\begin{equation}
    \lim_{v\searrow v_0}\inf_{S_v}|\slashed g(v)|_{\hat{\slashed g}}=0.\label{eq:continuation}
\end{equation}
\end{lem}

 A \emph{conformal class of Riemannian metrics} on $C$ is the equivalence class $\mathfrak K$ of symmetric $S$-$(0,2)$-tensors on $C$ which are positive definite on each $S_v$ with the equivalence relation ${\slashed g}{}',{\slashed g}''\in \mathfrak K$ if there exists a smooth positive function $\psi$ on $C$ such that ${\slashed g}'=\psi^2{\slashed g}''$.

\cref{lem:seed-data} shows that the free data\footnote{That is, the quantities that may be freely prescribed.} for the characteristic data (in the gauge $\Omega^2=1$) are given by 
\begin{equation*}
    \mathfrak K, \slashed g(1), \tr\chi(1),\tr\underline\chi(1), \eta(1), \hat{\underline\chi}(1), \text{and }\underline\alpha(1),
\end{equation*}
subject to the condition that $\slashed g(1)$ is compatible with $\mathfrak K$ and that $\hat{\underline\chi}(1)$ and $\underline\alpha(1)$ are traceless, which is a notion that depends only on $\mathfrak K$. The desired induced metric $\slashed g$ will be a representative of $\mathfrak K$. One often writes $\mathfrak K=[\slashed g]$, so the prescription of $\mathfrak K$ is the prescription of the conformal class of the induced metric $\slashed g$ which is to be found. 

The condition \eqref{eq:traceless} on the representative $\hat{\slashed g}$ of $\mathfrak K$ can be imposed without loss of generality, i.e., $\mathfrak K$ always contains a representative satisfying \eqref{eq:traceless}. Indeed, let $\tilde{\slashed g}\in[\slashed g]$ and let $\psi\doteq \exp(\int_v^1 \tfrac 14\tr_{\tilde{\slashed g}}D\tilde{\slashed g}\,dv').$ Then $\hat{\slashed g}\doteq \psi^2\tilde{\slashed g}$ satisfies \eqref{eq:traceless}. 

\begin{rk}
In \cref{lem:mathfrak-h} below, we will directly construct a specific $\hat{\slashed g}$ satsifying the volume form condition $D(d\mu_{\hat{\slashed g}})=0$, which easily implies \eqref{eq:traceless} by the first variation formula for area.
\end{rk}

\begin{proof}[Outline of the proof of \cref{lem:seed-data}]
    Let $\phi_1$ be the positive function on $S_1$ satisfying $\slashed g{}_1=\phi^2_1\hat{\slashed g}{}_1$. We make the ansatz
    \begin{equation}
    \slashed g = \phi^2 \hat{\slashed g}, \label{eq:phi-defn}
\end{equation} on $C$, where $\phi$ is now a positive function on $C$ agreeing with $\phi_1$ on $S_1$. 
    
 We define
\begin{equation}
    e \doteq \tfrac 18 \hat{\slashed g}{}^{AB} \hat{\slashed g}{}^{CD} \partial_v \hat{\slashed g}{}_{AC}\partial_v\hat{\slashed g}{}_{BD}\label{eq:e-calc}
\end{equation}
     relative to any Lie-transported angular coordinate system on the spheres. We set
\begin{equation}
    \partial_v\phi_1\doteq 2\phi_1\tr\chi_1
\end{equation}
and let $\phi$ be the unique solution of the second order ODE 
\begin{equation}
    \partial_v^2\phi + e\phi =0, \label{eq:Ray-phi}
\end{equation}
with initial conditions $(\phi_1,\partial_v\phi_1)$. If $\phi$ remains strictly positive on all of $C$, then we let $v_0$ be any strictly negative number. If however $\phi$ has a zero on $C$, then we take $v_0$ to be the supremum of $v\in [0,1]$ for which $\inf_{S_v}\phi \le 0$. This definition gives \eqref{eq:continuation}.

We now set
\begin{align}
    \hat\chi &\doteq \tfrac 12 \phi^2 D\hat{\slashed g}\label{eq:first-var-seed}\quad\text{and}\\
       \tr\chi &\doteq \tfrac 12\partial_v \log\phi\label{eq:tr-chi-phi}
\end{align}
along $C$ and observe that this choice of $\tr\chi$ is consistent with \eqref{eq:x(1)}. By \eqref{eq:traceless}, the shear defined by \eqref{eq:first-var-seed} is $\slashed g$-traceless. From \eqref{eq:phi-defn}, we have
\[D\slashed g = \phi^2 D\hat{\slashed g}+2\phi \partial_v\phi \hat{\slashed g}=\phi^2D\hat{\slashed g}+2\partial_v\log\phi\,\slashed g,\]
and by comparing with the first variation formula \eqref{eq:first-variation-v} written in the form
\[D\slashed g = 2\hat\chi + \tr\chi\,\slashed g,\]
we conclude that \eqref{eq:first-var-seed} and \eqref{eq:tr-chi-phi} are consistent with the first variation formula. 

From \eqref{eq:first-var-seed}, we have
\begin{equation}
    e = \tfrac 12|\hat\chi|^2,\label{eq:e-defn}
\end{equation}
so the ODE \eqref{eq:Ray-phi} is seen to be equivalent to  Raychaudhuri's equation \eqref{eq:Ray-v}.

From here, the full null data along $C$ can be determined by stepping through the null structure and Bianchi equations in the right order, as in \cite{Christo09}. We will outline this procedure in the proof of \cref{lem:bottom-sphere-estimates} below.
\end{proof}

\subsection{Reference sphere data for the Kerr family} \label{sec:ref-sphere-data}

\begin{defn}[The Kerr family of metrics]
    Let $\mathcal M_*\doteq(-\infty,\infty)_v\times (0,\infty)_r\times S^2$, where $S^2$ carries standard spherical polar coordinates $\vartheta$ and $\varphi$. The \emph{Kerr family} of metrics is the smooth two-parameter family of Lorentzian metrics 
     \begin{equation}
        g_{M,a} = -\left(1-\frac{2Mr}{\Sigma}\right) dv^2 + 2\,dv\,dr - \frac{4Mar\sin^2\vartheta}{\Sigma}dv\,d\varphi -2a\sin^2\vartheta \,dr\,d\varphi + \Sigma\,d\vartheta^2 + \rho^2 \sin^2\vartheta\,d\varphi^2
        \label{eq:Kerr-2}
    \end{equation}
    on $\mathcal M_*$, where $M\ge 0$ is the mass, $a\in \Bbb R$ is the specific angular momentum, 
    \begin{align*}
        \Sigma&\doteq r^2+a^2\cos^2\vartheta,\quad\text{and}\\
        \rho^2&\doteq r^2+a^2 +\frac{2Ma^2r\sin^2\vartheta}{\Sigma}.
    \end{align*} When $a=0$, $g_{M,a}$ reduces to the \emph{Schwarzschild metric} 
    \begin{equation}
    g_M= -\left(1-\frac{2M}{r}\right) dv^2 +2\,dv\,dr + r^2 \gamma,
\end{equation} where $\gamma\doteq d\vartheta^2+\sin^2\vartheta\,d\varphi^2$. 
When $M=0$, $g_{M,a}$ reduces to the \emph{Minkowski metric}
\begin{equation*}
    m\doteq -dv^2+2\,dv\,dr+r^2\gamma.
\end{equation*}
\end{defn}

The metrics $g_{M,a}$ solve the Einstein vacuum equations \eqref{eq:EVE}. The spacetime $(\mathcal M_*,g_{M,a})$ is time-oriented by $\partial_v$ for $r\gg 1$. The vector field $\partial_v$ is Killing---the Kerr family is stationary. If $|a|\le M$ and $M>0$, these metrics describe black hole spacetimes. For $0\le |a|<M$, the black hole is said to be \emph{subextremal}, and for $0< |a|=M$, \emph{extremal}.

\begin{rk}
In the context of the Schwarzschild solution, the coordinates $(v,r,\vartheta,\varphi)$ are called \emph{ingoing Eddington--Finkelstein coordinates}. Indeed, defining 
\begin{equation*}
    t\doteq v-r-2M\log|r-2M|
\end{equation*}
brings $g_M$ into the familiar form 
\begin{equation*}
    g_M= -\left(1-\frac{2M}{r}\right)dt^2+\left(1-\frac{2M}{r}\right)^{-1}dr^2+r^2\gamma
\end{equation*} and $(t,r,\vartheta,\varphi)$ are called \emph{Schwarzschild coordinates}. 
    In the context of the Kerr solution, the coordinates $(v,r,\vartheta,\varphi)$ are called \emph{Kerr-star coordinates}. For the relation to the perhaps more familiar \emph{Boyer--Lindquist coordinates}, see \cite{ONeill-book}. The advantage of defining the Kerr family $g_{M,a}$ directly in these coordinates is that we may view it as a smooth two-parameter family of Lorentzian metrics on the \emph{fixed} smooth manifold $\mathcal M_*$, even across the horizons located at $r_\pm = M\pm\sqrt{M^2-a^2}$ when $M>0$. 
\end{rk}

\begin{rk}
    The spacetimes $(\mathcal M_*,g_{M,a})$ defined here do not cover the entire maximal analytic extensions of the Minkowski, Schwarzschild, and Kerr solutions. Most importantly, $(\mathcal M_*,g_{M,a})$ includes the portion of the future event horizon $\mathcal H^+\doteq \{r=r_+\}$ strictly to the future of the bifurcation sphere. 
\end{rk}

We will now define the reference sphere data for the Kerr family. We will use the notion of sphere data $x[g,i,\underline L]$ generated by a Lorentzian metric $g$ on a smooth manifold $\mathcal M$, an embedding $i:S^2\to\mathcal M$, and a choice of null vector field $\underline L$ defined along and orthogonal to $i(S^2)$, which is defined in \cref{lem:data-generation} below. Note that
\begin{equation*}
    Y\doteq -\partial_r
\end{equation*}
is a future-directed null vector field for $(\mathcal M_*,g_{M,a})$. We also define the family of embeddings 
\begin{align*}
    i_R:S^2&\to\mathcal M_*\\
   (\vartheta,\varphi) &\mapsto (0,R,\vartheta,\varphi)
\end{align*}
for $R>0$. 
\begin{defn}[Reference sphere data]
    Let $M\ge 0$, $a\in\Bbb R$, $R>0$, and $m\ge 0$ be an integer. The \emph{reference Kerr sphere data set} of mass $M$, specific angular momentum $a$, and radius\footnote{We use the term radius because it is associated to the Kerr coordinate $r$, but the spheres are not round!} $R$ is the $C^2_uC^{2+m}_v$ sphere data set given by 
    \begin{equation}
        \mathfrak k_{M,a,R}\doteq x[g_{M,a},i_R,Y].
    \end{equation}
    The \emph{reference Schwarzschild data sets} are defined by 
    \begin{equation}
        \mathfrak s_{M,R}\doteq \mathfrak k_{M,0,R}
    \end{equation}
    and the \emph{reference Minkowski data sets} are defined by 
    \begin{equation}
        \mathfrak m_R\doteq \mathfrak s_{0,R}. 
    \end{equation}
\end{defn}
We will colloquially refer to $\mathfrak k_{M,a,R}$ as a ``Kerr coordinate sphere'' and $\mathfrak s_{M,R}$ (resp., $\mathfrak m_R$) as a ``(round)~Schwarzschild symmetry sphere'' (resp., ``(round) Minkowski symmetry sphere''). 

In the notation of \cref{sec:sphere-data}, one can show that
    \begin{align}
        {\mathfrak s}_{M,R}^\mathrm{low}&= \left(1,R^2\gamma,\frac{2}{R}\left(1-\frac{2M}{R}\right),0,-\frac{2}{R},0,\dotsc,0\right), \label{eq:Schw-sphere-data}\\
        {\mathfrak s}_{M,R}^\mathrm{high}&= 0\label{eq:Schw-sphere-data-high}.
    \end{align}
    A similarly simple expression is neither available nor needed for Kerr. Indeed, we have the 
    \begin{lem}
For any integer $m\ge 0$, $\mathfrak k_{M,a,R}$ is a smooth three-parameter family of $C^2_uC^{2+m}_v$ sphere data sets. In particular,
            \begin{equation}
   \lim_{a\to 0} \|\mathfrak k_{M,a,R}-\mathfrak s_{M,R}\|_{\mathcal X^m} =0.\label{eq:Kerr-approx}
\end{equation}
    \end{lem}
    \begin{proof}
        The metrics $g_{M,a}$ are defined on the fixed smooth manifold $\mathcal M_*$. By inspection of \eqref{eq:Kerr-2}, $g_{M,a}$ varies smoothly in $M$ and $a$. Therefore, the smooth dependence of $\mathfrak k_{M,a,R}$ on the parameters and \eqref{eq:Kerr-approx} follow from the smooth dependence of the sphere data generated by $(g,i,\underline L)$ on $g$, $i$, and $\underline L$; see \cref{lem:data-generation}. 
    \end{proof}

We conclude this section with several remarks.

\begin{rk}
   As was already mentioned, the Kerr family is stationary. Defining $i_R(\vartheta,\varphi)=(v,R,\vartheta,\varphi)$ for any $v\in\Bbb R$ leads to the same sphere data. 
\end{rk}

\begin{rk}
    We always take the Kerr axis to point along the poles of the fixed identification of the Kerr coordinate spheres with the usual unit sphere. 
\end{rk}

\begin{rk}
    \label{rk:BL-conformal-factor} The induced metric $\slashed g{}_{M,a,R}$ in $\mathfrak k_{M,a,R}$ is not conformal to the round metric $\gamma$ (defined relative to the Kerr angular coordinates). For this reason we have slightly modified the setup in \cref{sec:Chr-seed-data} by imposing \eqref{eq:traceless} instead of simply $d\mu_{\hat{\slashed g}}=d\mu_\gamma$ as in \cite[Chapter 2]{Christo09}. See already \cref{lem:mathfrak-h} below. 
\end{rk}

    \begin{rk}
      The induced metric  $\slashed g{}_{M,a,R}$ is given in Kerr angular coordinates by 
        \begin{equation}
            \slashed{g}_{M,a,R}= \Sigma\,d\vartheta^2+ \rho^2\sin^2\vartheta \,d\varphi^2.
        \end{equation}
            To show that this extends smoothly over the poles relative to the smooth structure defined by the Kerr angular coordinates, we note the identity 
            \begin{equation}
                \Sigma\,d\vartheta^2+ \rho^2\sin^2\vartheta \,d\varphi^2 = \Sigma(d\vartheta^2+\sin^2\vartheta\,d\varphi^2)+a^2\left(1+\frac{2Mr}{\Sigma}\right)\sin^4\vartheta\,d\varphi^2. \label{eq:smooth-ext}
            \end{equation}
            Now $\sin^2\vartheta\,d\varphi$ is a globally defined smooth 1-form on $S^2$ since it is the $\gamma$-dual of the globally defined vector field $\partial_\varphi$, so the right-hand side of \eqref{eq:smooth-ext} can be extended smoothly over the poles. 
    \end{rk}

\subsection{Perturbative characteristic gluing}

Since the characteristic gluing results of \cite{ACR1,ACR2,ACR3,Czimek2022-cl} pass through linear theory, the conserved charges in Minkowski space play an important role. In \cref{sec:conserved-charges}, we give the definition of conserved charges. In \cref{sec:CR}, we state the main result of \cite{Czimek2022-cl} in the form which we will directly apply it. 

\subsubsection{Conserved charges}
\label{sec:conserved-charges}

\begin{defn}[Spherical harmonics]
    For $\ell\in\Bbb N_{0}$ and $m=-\ell,\dotsc,\ell$, let $Y^\ell_m$ denote the standard real-valued spherical harmonics on the unit sphere $(S^2,\gamma)$. We also define the electric and magnetic 1-form spherical harmonics by 
    \begin{equation*}
        E^\ell_m\doteq  -\frac{1}{\sqrt{\ell(\ell+1)}}\slashed\nabla Y^\ell_m\quad \text{and}\quad H^\ell_m\doteq \frac{1}{\sqrt{\ell(\ell+1)}}{}^*\slashed\nabla Y^\ell_m
    \end{equation*}
    for $\ell\ge 1$ and $|m|\le \ell$. By a standard abuse of notation, we will use the same symbol for the vector-valued spherical harmonics, with the understanding that $\gamma$ is used to raise the index. 
\end{defn}

\begin{defn}[Linearly conserved charges]\label{def:charges} Let $x$ be $C^2$ sphere data and define the 1-form $\mathbf B$ and scalar function $\mathbf m$ by
    \begin{align*} \mathbf B&\doteq \frac{\phi^3}{2\Omega^2}\left(\slashed \nabla(\Omega\tr\chi)+\Omega\tr\chi (\eta-2\slashed\nabla\log\Omega)\right) \text{ and}\\
    \mathbf m&\doteq \phi^3\left(K+\tfrac 14 \tr\chi\tr\underline\chi\right)-\phi\adiv\mathbf B.
    \end{align*} The conformal factor $\phi$ is defined as the unique positive function on $S^2$ such that $d\mu_{\slashed g}=\phi^2 d\mu_\gamma$, where $\gamma$ is the distinguished choice of round metric on $S$.
    Then the charges $\mathbf E, \mathbf P,\mathbf L$, and $\mathbf G$ (where the latter three are vectors in $\Bbb R^3$ indexed by $m\in\{-1,0,1\}$) are defined by
    \begin{align*}
    \mathbf E &\doteq \mathbf m^{\ell = 0}\\
    \mathbf P&\doteq \mathbf m^{\ell = 1}\\
    \mathbf L&\doteq \mathbf B^{\ell =1,H}\\
    \mathbf G&\doteq\mathbf B^{\ell =1,E}.
    \end{align*}
    Here the modes are defined by 
    \begin{align*}
       f^{\ell =0} &\doteq \int_{S^2}f \, Y^{0}_0 \,d\mu_\gamma,& (f^{\ell=1})^m &\doteq \int_{S^2}f \, Y^{1}_m \,d\mu_\gamma,\\
       ( X^{\ell = 1, E})^m &\doteq \int_{S^2}\gamma(X,E^1_m)\,d\mu_\gamma,& ( X^{\ell = 1, H})^m &\doteq \int_{S^2}\gamma(X,H^1_m)\,d\mu_\gamma.
    \end{align*}
\end{defn}

\subsubsection{Czimek--Rodnianski obstruction-free perturbative characteristic gluing}
\label{sec:CR}

The following theorem is a combination of \cite[Theorem 3.2]{ACR2}, \cite[Theorem 2.9]{Czimek2022-cl}, and Remark (5) after Theorem 2.9 in \cite{Czimek2022-cl}. 

\begin{thm}[Czimek--Rodnianski obstruction-free characteristic gluing]\label{thm:CR}
For any $C_\mathbf E>0$ and integer $m\ge 0$, there exist constants $C_*>0$ and $\ve_0>0$ such that the following holds. Let $\underline x$ be ingoing null data on an ingoing cone $\underline C=[-\frac{1}{100},\frac{1}{100}]_u\times S^2$ solving the null structure equations and Bianchi identities, and $x_2$ be $C^2_u C^{2+m}_{v}$ sphere data. Let $x_1$ be the sphere data in $\underline x$ corresponding to $u=0$. Let 
\begin{equation*}
    (\Delta \mathbf E,\Delta\mathbf P,\Delta\mathbf L,\Delta\mathbf G)\doteq (\mathbf E,\mathbf P,\mathbf L,\mathbf G)(x_2)- (\mathbf E,\mathbf P,\mathbf L,\mathbf G)(x_1)
\end{equation*}
be the difference of the conserved charges of $x_2$ and $x_1$. If the data sets satisfy the smallness condition
\begin{equation}
    \|\underline x-\mathfrak m\|_{\mathcal X^+(\underline C)}+\|x_2-\mathfrak m_2\|_{\mathcal X^m}<\ve_\flat\label{eq:CR-cond-4}
\end{equation} 
for some $0<\ve_\flat<\ve_0$, where $\mathfrak m$ is reference Minkowski null data\footnote{That is, reference Minkowski sphere data defined along the ingoing cone $\underline{C}$. See \cite{ACR2} for details.} and $\mathfrak m_2$ is reference Minkowski sphere data, and the following ``coercivity'' conditions on the charge differences
\begin{align}
\label{eq:CR-cond-1}    \Delta\mathbf E &> C_\mathbf E \ve_\flat, \\
 \label{eq:CR-cond-2}     |\Delta\mathbf L|&< \ve^2_\flat, \text{and} \\
  \label{eq:CR-cond-3}    |\Delta \mathbf P|+|\Delta\mathbf G| &< C_*\Delta \mathbf E,
\end{align}
then there is a solution $x\in  \mathcal X(C)$ of the null structure equations along a null hypersurface $C=[1,2]_v\times S^2$ such that 
\begin{equation}
    x(1)=x_1' ,\quad x(2)=x_2,\label{eq:matching-1}
\end{equation}
and 
\begin{equation*}
    \|x-\mathfrak m\|_{\mathcal X^{\mathrm{h.f.}}(C)}+\| x_1'-\mathfrak m_1\|_{\mathcal X}\les \|\underline x-\mathfrak m\|_{\mathcal X^+(\underline C)}+\|x_2-\mathfrak m_2\|_{\mathcal X^m}.
\end{equation*}
The sphere data $x_1'$ is 
obtained by applying a sphere diffeomorphism and a transversal sphere perturbation to $x_1$ inside of $\underline C$. See \cite{ACR2, Czimek2022-cl} for the precise definitions of these terms.
\end{thm}

\begin{rk}
    The matching condition \eqref{eq:matching-1} is to order $C^{2+m}$ in directions tangent to the cone. 
\end{rk}

Since all hypotheses in \cref{thm:CR} are open conditions, we immediately have:

\begin{cor}\label{cor:CR}
    If the sphere data set $x_2$ satisfies the hypotheses of \cref{thm:CR}, there exists an $\ve_*>0$ such that if $\tilde x_2$ is another sphere data set such that 
    \[\|\tilde x_2-x_2\|_{\mathcal X^m} <\ve\]
    for some $0\le \ve < \ve_*$, then the conclusion of the theorem holds for $\tilde x_2$ in place of $x_2$.
\end{cor}

\section{Proofs of the main gluing theorems}

\subsection{Gluing an almost-Schwarzschild sphere to a round Schwarzschild sphere with a larger mass}

In this subsection, we prove the main technical lemma of our paper. In essence, we show how to decrease the mass of a Schwarzschild sphere (going backwards in time) by an arbitrary amount, with an arbitrary small error. 

\begin{prop}\label{prop:Schw-perturb-Schw} Given any $0 \leq M_* < M$, $R>0$, integer $m\ge 0$, and any  $\ve_\sharp >0$, there exists a $\delta>0$ and null data $x$ on $C_1^{[0,1]}=[0,1]\times S^2$ solving the null structure equations and Bianchi identities such that 
\begin{equation}
    \mathfrak b_\delta(x(1))= \mathfrak s_{M,R}\label{eq:boost-2}
\end{equation}
and 
\begin{equation}
  \label{eq:boost-est-1}  \|\mathfrak b_\delta(x(0))-\mathfrak s_{M_*,R}\|_{\mathcal X^m} <\ve_\sharp,
\end{equation}
where $\mathfrak b$ is the boost operation defined in \cref{def:boost} and $\mathcal X^m$ is the sphere data norm appearing in \cref{thm:CR}. 
\end{prop}

The proof of the proposition is given at the end of this subsection. We first give a general construction of seed data $\hat{\slashed g}$ compatible with the hypotheses of \cref{lem:seed-data}.

\begin{lem}\label{lem:mathfrak-h}
    Let $C$ be as in \cref{sec:Chr-seed-data}. Let $\tilde \gamma$ be a Riemannian metric on $S^2$. There exists an explicitly defined smooth assignment $\tilde\gamma\mapsto \mathfrak h$, where $\mathfrak h$ is a traceless $(1,1)$-S-tensor field along $C$, such that for any $\lambda\in\Bbb R$,  \begin{equation}
         \hat{\slashed g}{}_{AB}\doteq \tilde{\gamma}{}_{AC}\exp(\lambda\mathfrak h)^C{}_B \label{defn:g-hat}
     \end{equation}
     defines a Riemannian metric for each $v$, satisfies condition \eqref{eq:traceless}, and $\hat{\slashed g}(1)=\tilde\gamma$. Here $\tilde\gamma$ is defined along $C$ according to $D\tilde\gamma=0$. We have
     \begin{equation}
    \partial_v \hat{\slashed g}{}_{AB} = \lambda \hat{\slashed g}{}_{AC} \partial_v \mathfrak h^C{}_B  \label{eq:D-g-slash-hat}
\end{equation}
and the inverse metric is given by 
\begin{equation}
    \hat{\slashed g}{}^{AB} = \left(\hat{\slashed g}^{-1}\right)^{AB}=({\tilde\gamma}^{-1})^{AC} \exp(-\lambda\mathfrak h)^B{}_C = ({\tilde\gamma}^{-1})^{BC} \exp(-\lambda\mathfrak h)^A{}_C.\label{eq:matrix-inverse}
\end{equation}
\end{lem}

\begin{proof}
    
We first fix some cutoff functions. Let $\chi\in C^\infty_c(0,\tfrac 12)$ be nonnegative and not identically zero. Let $\chi_1\doteq \chi(v)$ and $\chi_2(v) \doteq\chi(v-\tfrac 12)$. Let $p_1$ be the north pole of $S^2$, $p_2$ the south pole, and set $U_i\doteq S^2\setminus\{p_i\}$ for $i=1,2$. Let $f_1\in C^\infty_c(U_1)$ and $f_2\in C^\infty_c(U_2)$ be such that $f_1^2+f_2^2= 1$ on $S^2$. 

Let $(\vartheta_1^1,\vartheta_1^2)$ be a coordinate chart covering $U_1$, $(\vartheta_2^1,\vartheta_2^2)$ be a coordinate chart covering $U_2$, and set
\[\mathring h_i \doteq d\vartheta_i^1\otimes \frac{\partial}{\partial\vartheta_i^1}-d\vartheta_i^2\otimes \frac{\partial}{\partial\vartheta_i^2}\text{ on }U_i.\] As matrices, these tensor fields are given by $\diag(1,-1)$ in the respective coordinate systems. 

We now claim that the symmetric $(0,2)$-tensor fields 
\[ h_{i\,AB}\doteq \tfrac 12\left(\tilde\gamma{}_{AC}\mathring h_i^C{}_B+\tilde\gamma{}_{BC}\mathring h_i^C{}_A\right)\]
are nowhere vanishing on their respective domains of definition. 
This follows from the fact that 
\begin{equation*}
     h_{i\,AA} = \tilde\gamma\left(\frac{\partial}{\partial\vartheta_i^A},\frac{\partial}{\partial\vartheta_i^A}\right),
\end{equation*}
where no summation is implied. Since $\tilde\gamma$ is positive definite, we must at each point have both $ h_{i\,11}$ and $ h_{i\,22}$ nonvanishing, so $ h_i$ is always nonvanishing. Let $ h^\sharp_{i}$ be the $(1,1)$-tensor field obtained by dualizing $ h_i$ with $\tilde\gamma$. We then finally define
\begin{equation}
    \mathfrak h\doteq \chi_1 f_1| h_1|_{\tilde\gamma}^{-1}  h_1^\sharp + \chi_2 f_2| h_2|_{\tilde\gamma}^{-1} h_2^\sharp.\label{eq:h-defn}
\end{equation}
It is clear that $\tr\mathfrak h =0$ and that $\mathfrak h^\flat$ is symmetric, where $\flat$ is taken relative to $\tilde\gamma$. 

We now show that \eqref{defn:g-hat} defines a Riemannian metric. Viewing $\mathfrak h$ as an endomorphism $TS^2\to TS^2$, the power series 
\begin{equation}
    \exp(\lambda\mathfrak h) = \sum_{n=0}^\infty \frac{1}{n!}(\lambda \mathfrak h)^n\label{eq:power-series}
\end{equation}
converges and defines a smooth family of endomorphisms. 

To verify that $\hat{\slashed g}$ is symmetric, we examine \eqref{eq:power-series} term by term:
\begin{align*}
 \tilde\gamma{}_{AC}   \left((\lambda\mathfrak h)^n\right)^C{}_B &=\lambda^n \tilde\gamma{}_{AD_1} \mathfrak h^{D_1}{}_{D_2} \mathfrak h^{D_2}{}_{D_3}\cdots \mathfrak h^{D_{n-1}}{}_{D_n} \mathfrak h^{D_n}{}_B\\
 &=\lambda^n \tilde\gamma{}_{D_1D_2} \mathfrak h^{D_1}{}_{A} \mathfrak h^{D_2}{}_{D_3}\cdots \mathfrak h^{D_{n-1}}{}_{D_n} \mathfrak h^{D_n}{}_B\\
 &=\cdots \\
 &= \lambda^n \tilde\gamma{}_{D_{n-1}D_n}\mathfrak h^{D_1}{}_{D_2}\mathfrak h^{D_2}{}_{D_3}\cdots \mathfrak h^{D_{n-1}}{}_A \mathfrak h^{D_n}{}_B\\
 &=\lambda^n \tilde\gamma{}_{D_{n-1}D_n}\mathfrak h^{D_1}{}_{D_2}\mathfrak h^{D_2}{}_{D_3}\cdots \mathfrak h^{D_{n-1}}{}_B \mathfrak h^{D_n}{}_A\\
 &= \tilde\gamma{}_{BC}   \left((\lambda\mathfrak h)^n\right)^C{}_A,
\end{align*}
where we used the symmetry of $\mathfrak h^\flat$ repeatedly. 
That $\hat{\slashed g}$ is positive definite follows easily from the fact that at the origin of a normal coordinate system for $\tilde\gamma$, $\hat{\slashed g}{}_{AB}$ is the matrix exponential of a symmetric matrix, and hence positive definite. 

To show that \eqref{eq:traceless} is satisfied, we use Jacobi's formula to calculate 
\begin{equation*}
    \det \hat{\slashed g} = \det \tilde\gamma\, \exp(\lambda \tr\mathfrak h) = \det\tilde\gamma
\end{equation*}
relative to any coordinate system, where we used $\tr\mathfrak h =0$. We conclude that the volume form of $\hat{\slashed g}$ satisfies 
\begin{equation}
  d\mu_{\hat{\slashed g}} = d\mu_{\tilde\gamma}.  \label{eq:volume-form}
\end{equation}
Observe that since $D(d\mu_{\tilde\gamma})=0$ by definition of $\tilde{\gamma}$ along $C$, \eqref{eq:volume-form} implies
\[0=D(d\mu_{\hat{\slashed g}})= \tfrac 12 \tr_{\hat{\slashed g}}(D\hat{\slashed g})\,d\mu_{\hat{\slashed g}},\] so $D\hat{\slashed g}$ is $\hat{\slashed g}$-traceless.

To prove \eqref{eq:D-g-slash-hat}, we use the fact that $\mathfrak h(v)$ and $\mathfrak h(v')$ commute for any $v$ and $v'$ sufficiently close to simply differentiate \eqref{defn:g-hat}:
\begin{equation*}
    D\hat{\slashed g}{}_{AB}= \partial_v \hat{\slashed g}{}_{AB} = \gamma_{AC}\exp(\lambda \mathfrak h)^C{}_D \lambda \partial_v \mathfrak h^D{}_B = \lambda \hat{\slashed g}{}_{AC} \partial_v \mathfrak h^C{}_B.  
\end{equation*}
The formula \eqref{eq:matrix-inverse} is immediately seen to hold. 
\end{proof}

\begin{rk}
 By the Poincar\'e--Hopf theorem, the shear $\hat\chi$ must vanish at some point on each $S_v\subset C$. Equivalently, any $\mathfrak h$ for which \eqref{defn:g-hat} satisfies condition \eqref{eq:traceless}, must vanish at some point on each $S_v$. In order to satisfy \eqref{eq:average} below, this zero cannot stay along the same generator of $C$. The simplest solution to this problem is the two-pulse configuration above.  
\end{rk}

With this general construction out of the way, we begin the proof of \cref{prop:Schw-perturb-Schw} in earnest. We specialize now to the case of $\tilde\gamma =\gamma$, the round metric on the unit sphere.

\textbf{Convention.} We now introduce a parameter $\delta>0$ satisfying $0< \delta < \delta_0$, where $\delta_0>0$ is a sufficiently small fixed parameter only depending on $M_*, M, R$ and the fixed choices of $\chi, U_1,U_2, f_1$, and $f_2$. We will further use in this section the notation that implicit constants in $\lesssim, \gtrsim$, and $ \approx$ may depend $M_*, M, R$ and $\chi, U_1,U_2, f_1$, and $f_2$.  We also use the notation $\lesssim_j, \gtrsim_j$, and $ \approx_j$ if the implicit constants in $\lesssim, \gtrsim$, and $ \approx$ depend on an additional parameter $j$.

\begin{lem}
The geometric quantity $e$, defined in \eqref{eq:e-defn}, satisfies 
\begin{equation}
    \slashed\nabla \int_0^1 e\,dv=0 \label{eq:average}
\end{equation}
and 
\begin{equation}
    |\slashed \nabla^j e|\les_j \lambda^2. \label{eq:e-angular-estimate}
\end{equation}
for $j\geq 0$. 
\end{lem}
\begin{proof}
We have
\begin{equation}
e = \tfrac 18 \lambda^2 \partial_v \mathfrak h^A{}_B\partial_v\mathfrak h^B{}_A = \tfrac 18 \lambda^2 \left((\chi_1')^2 f_1^2 + (\chi_2')^2 f_2^2\right)\nonumber
\end{equation}
by \eqref{eq:e-calc}, \eqref{eq:matrix-inverse}, and \eqref{eq:h-defn}. Therefore, we have
\begin{equation}
\slashed \nabla^j e = \tfrac 18 \lambda^2 \left((\chi_1')^2 \slashed\nabla^j f_1^2 +(\chi_2')^2\slashed\nabla^j f_2^2\right),\nonumber
\end{equation}
which immediately proves \eqref{eq:e-angular-estimate}. To prove \eqref{eq:average}, we note that 
\begin{equation}
\int_0^1 \left((\chi_1')^2 f_1^2 + (\chi_2')^2 f_2^2\right)\,dv = \left(f_1^2 +f_2^2\right)\int_0^1 (\chi')^2\,dv=\int_0^1 (\chi')^2\,dv,\nonumber
\end{equation}
which is independent of the angle on $S^2$. 
\end{proof}

Along $[0,1]\times S^2$ we impose the gauge condition $\Omega^2=1$ and at $v=1$ we impose
\begin{align}
  \label{eq:trchi1} \tr\chi(1) = \delta \frac{2}{R}\left(1-\frac{2M}{R}\right).
\end{align}
The conformal factor $\phi$ solves Raychaudhuri's equation \eqref{eq:Ray-phi} with final values (see \eqref{eq:tr-chi-phi} and \eqref{eq:trchi1})
\begin{align*}
  \phi(1)  &=R_f\\
   \partial_v\phi(1) &=  \delta \left(1-\frac{2M}{R}\right)
\end{align*}

\begin{lem}
If $0<\delta\leq \delta_0$ and   $0\le \lambda\le \delta^{1/4}$, then 
\begin{equation}
|\slashed\nabla^j(\phi-R)|+|\slashed\nabla^j \partial_v\phi| \les_j \delta +\lambda^2\label{eq:phi-est-1}
\end{equation}
uniformly on $[0,1]\times S^2$ for every integer $j\ge 0$. 
\end{lem}
\begin{proof}
Integrating \eqref{eq:Ray-phi}, we obtain 
\begin{equation}
\partial_v\phi(v) = \delta \left(1-\frac{2M}{R}\right) + \int_v^1 \phi\,e \, dv'.\label{eq:first-integral}
\end{equation}
Assuming $|\phi|\le 10 R_f$ in the context of a simple bootstrap argument, we see that \eqref{eq:first-integral} and \eqref{eq:e-angular-estimate} imply 
\begin{equation}
|\partial_v\phi|\les\delta + \lambda^2, \label{eq:A1}
\end{equation}
which implies 
\begin{equation}
|\phi - R|\les \delta + \lambda^2. \label{eq:A2}
\end{equation} Since $\lambda^2\le \delta^{1/2}$, taking $\delta_0>0$ sufficiently small closes the boostrap and \eqref{eq:A1} and \eqref{eq:A2} hold on $[0,1]$. Commuting \eqref{eq:Ray-phi} repeatedly with $\slashed\nabla$ and arguing inductively using \eqref{eq:A1} and \eqref{eq:A2} as the base cases, we easily obtain \eqref{eq:phi-est-1}. 
\end{proof}

\begin{lem} Fix an angle $\theta_0\in S^2$. 
For $0< \delta \leq \delta_0$, the function 
\begin{equation}\nonumber
\lambda\longmapsto \tr \chi(0,\theta_0;\lambda)
\end{equation}
is monotonically increasing. 
\end{lem}
\begin{proof}
Since $\Omega=1$ identically, Raychaudhuri's equation \eqref{eq:Ray-v} becomes
\begin{equation}
\partial_v \tr\chi = -2\lambda^2e_1 - \tfrac 12 (\tr\chi)^2,\label{eq:Ray-mod-1}
\end{equation}
where $e_1\doteq \tfrac 18 \partial_v \mathfrak h^A{}_B \partial_v\mathfrak h^B{}_A$. 
Taking the $\partial_\lambda$ derivative of \eqref{eq:Ray-mod-1} gives
\begin{equation}
\partial_v\left(\partial_\lambda \tr\chi \right) = -4\lambda e_1 - \tr\chi(\partial_\lambda\tr\chi)\nonumber
\end{equation}
This is at once solved for 
\begin{equation}\nonumber
\partial_\lambda \tr\chi(v) = 4\lambda  e^{\int_v^1 \tr\chi\,dv'}\int_v^1 e^{-\int_{v'}^1\tr\chi\,dv''}e_1\,dv',
\end{equation}
which is strictly positive at $v=0$ for $\lambda>0$.
\end{proof}

\begin{lem}
   Let $0< \delta \leq \delta_0$ and  $0\le \lambda\le \delta^{1/4}$. Then $\tr\chi$ is monotonically decreasing along each generator and
    \begin{equation}
      \inf_{[0,\tfrac 12]\times S^2}  \tr\chi \gtrsim \delta +\lambda^2.\label{eq:tr-chi-lower-bnd-1}
    \end{equation}
\end{lem}
\begin{proof} Monotonicity of $\tr\chi$ follows at once from Raychaudhuri's equation \eqref{eq:Ray-mod-1}. 
    We can immediately integrate \eqref{eq:Ray-mod-1} to obtain 
    \begin{equation}
        \tr\chi(v)=\delta\frac{2}{R}\left(1-\frac{2M}{R}\right) +2 \lambda^2 e^{\tfrac 12\int_v^1\tr\chi \,dv'}\int_{v'}^1e^{-\tfrac 12\int_{v'}^1\tr\chi\,dv''}e_1dv', \label{eq:tr-chi-lower-bnd-2}
    \end{equation}
    By \eqref{eq:phi-est-1}, $\tr\chi \les \delta+\lambda^2\les \delta^{1/2}\le 1$, so \eqref{eq:tr-chi-lower-bnd-2} implies \eqref{eq:tr-chi-lower-bnd-1}.
\end{proof}

\begin{lem}\label{lem:construction-of-lambda} Fix an angle $\theta_0 \in S^2$.
For $0<\delta\leq \delta_0$, there exists a unique $\lambda_0=\lambda_0(\delta)\in (0,\delta^{1/4})$ (depending also on $\theta_0$) such that 
\begin{equation}
    \tr\chi(0,\theta_0;\lambda_0) = \delta\frac{2}{R}\left(1-\frac{2M_*}{R}\right),\label{eq:tr-chi-target}
\end{equation}
which also satisfies 
\begin{equation}
    \lambda_0(\delta) \approx \delta^{1/2}.\label{eq:lambda-0-est}
\end{equation}
\end{lem}
\begin{proof}
Let 
\begin{equation}
     c\doteq \frac{4}{R}(M-M_*)\quad\text{and}\quad
    C\doteq 2\int_0^1 e_1(v,\vartheta_0)\,dv. \nonumber
\end{equation}
Then the condition \eqref{eq:tr-chi-target} becomes (see \eqref{eq:trchi1}  and \eqref{eq:Ray-mod-1})
\begin{equation}
    c\delta = C\lambda^2 +\tfrac 12 \int_0^1(\tr\chi)^2\,dv.\label{eq:lambda-condition}
\end{equation}
Combining \eqref{eq:phi-est-1} with \eqref{eq:tr-chi-lower-bnd-1} shows immediately that \eqref{eq:lambda-condition} can be achieved by a $\lambda_0(\delta)$ satisfying \eqref{eq:lambda-0-est}.
\end{proof}

From now on, we always take $\lambda = \lambda_0(\delta)$ as constructed in Lemma~\ref{lem:construction-of-lambda}.  
With this $\lambda$, our main estimates \eqref{eq:phi-est-1} are improved to:
\begin{equation}
    |\slashed \nabla^j e|+ |\slashed\nabla^j(\phi-R)|+|\slashed\nabla^j \tr\chi| + \delta^{1/2}|\slashed\nabla^j \hat\chi| \les_j \delta \label{eq:improved}
\end{equation}
for any $j\ge 0$ and uniformly on $[0,1]\times S^2$. Importantly, we also have 

\begin{lem}
    Let $0<\delta\leq \delta_0$. Then,
   \begin{equation}
       |\slashed\nabla^j \tr\chi(0)|\les_j \delta^2\label{eq:tr-chi-0}
   \end{equation}
   at $v=0$ for $j\ge 1$. Hence,
   \begin{equation}
       \left|\slashed\nabla^j\left(\tr\chi(0)-\delta\frac{2}{R}\left(1-\frac{2M_*}{R}\right)\right)\right|\les_j \delta^2\label{eq:tr-chi-0-improved}
   \end{equation}
   for all $j\ge 0$ at $v=0$.
\end{lem}
\begin{proof}
Applying $\slashed\nabla^j$ to \eqref{eq:Ray-mod-1}, integrating in $v$, and applying \eqref{eq:average} yields 
\begin{equation*}\slashed\nabla^ j \tr\chi(0) = -\tfrac 12 \int_0^1 \slashed \nabla^j (\tr\chi)^2\,dv.\end{equation*}
We arrive at \eqref{eq:tr-chi-0} after applying \eqref{eq:improved}. This also proves \eqref{eq:tr-chi-0-improved} for $j\ge 1$. For $j=0$, we integrate \eqref{eq:tr-chi-0} along geodesics emanating from $\theta_0$ and use \eqref{eq:tr-chi-target}.
\end{proof} 

The remaining sphere data at $v=1$ is now specified as follows:
\begin{align}
\nonumber \tr\underline\chi(1) &= -\frac{1}{\delta}\frac{2}{R}\\
\nonumber\hat{\underline\chi}(1)&=0\\
\nonumber\eta(1)&=0\\
\nonumber\underline\omega(1) =\underline D\underline\omega(1)& = 0\\
\nonumber\alpha(1) =\underline\alpha(1)&=0.
\end{align}
Combining everything and using the null structure and Bianchi equations to solve the rest of the system, we have 

\begin{lem}\label{lem:bottom-sphere-estimates}
For $0<\delta\leq \delta_0$  we have \underline{at $v=0$}  
\begin{align}
\nonumber \delta^{-1}|\slashed\nabla^j(\slashed g-R^2\gamma)|+\delta^{-2}\left|\slashed\nabla^j\left(\tr\chi(0)-\delta\frac{2}{R}\left(1-\frac{2M_*}{R}\right)\right)\right|+\delta^{-1} \left|\slashed\nabla^j\left(K-\frac{1}{R^2}\right)\right|\\
 \nonumber  {}  + \delta^{-1/2} |\slashed\nabla^j \eta|+ \left|\slashed\nabla^j \left(\tr\underline\chi +\frac{1}{\delta}\frac{2}{R}\right)\right|+ \delta^{1/2}|\slashed\nabla^j \hat{\underline\chi}|+\delta^{-3/2}|\slashed\nabla^j\beta|&\\
  {} + \delta^{-1/2}\left|\slashed\nabla^j\left(\rho +\frac{2M_*}{R^3}\right)\right|+\delta^{-1/2}|\slashed\nabla^j\sigma|+|\slashed\nabla^j \underline\beta|+\delta^{1/2}|\slashed\nabla^j\underline\alpha|+|\slashed\nabla^j\underline\omega|+\delta |\slashed\nabla^j\underline D\underline\omega|&\les_j 1\label{eq:bottom-sphere-estimates}
\end{align}
 for every $j\ge 0$
and 
\begin{equation}
   \hat\chi(0)=0,\qquad  \alpha(0)=0.\label{eq:trivial-2}
\end{equation}
The terms in \eqref{eq:bottom-sphere-estimates} are displayed in the order in which they are estimated. 
\end{lem}
\begin{proof}
The proof follows the procedure of \cite[Chapter 2]{Christo09}, which we now outline. The first term is estimated using \eqref{eq:improved}. The second term was estimated in \eqref{eq:tr-chi-0-improved}. The third term is estimated using the formula 
\begin{equation*}
    K_{\slashed g} = \phi^{-2}(K_{\hat{\slashed g}}-\slashed\Delta_{\hat{\slashed g}} \log\phi).
\end{equation*}
Note that the first and third terms are estimated by $\delta^{1/2}$ on the whole cone, but are improved at $v=0$. To estimate the fourth term, the transport equation \eqref{eq:eta-transport} combined with the Codazzi equation \eqref{eq:Codazzi1} and \eqref{eq:eta-etabar} yields
\begin{equation*}
    \partial_v \eta_A+\tr\chi\,\eta_A = (\adiv\chi)_A-\slashed\nabla_A\tr\chi.
\end{equation*}
Now $|\eta|$ can be estimated using Gr\"onwall's inequality and \eqref{eq:improved}. To estimate the fifth term, the transport equation \eqref{eq:tr-chi-bar-transport} is combined with the Gauss equation \eqref{eq:Gauss-eqn} to give
\begin{equation*}
    \partial_v\tr\underline\chi + \tr\chi\,\tr\underline\chi = -2K-2\adiv\eta +|\eta|^2.
\end{equation*}
Gr\"onwall gives $|\!\tr\underline\chi|\les\delta^{-1}$, which then easily implies the desired estimate by Gr\"onwall applied to 
\[\left(\partial_v +\tr\chi\right)\left(\tr\underline \chi +\frac{1}{\delta}\frac{2}{R}\right)=-2K-2\adiv\eta +|\eta|^2+\frac{1}{\delta}\frac{2}{R}\tr\underline\chi.\]
To estimate the sixth term, apply Gr\"onwall directly to \eqref{eq:chi-bar-hat-transport}. The first variation formula \eqref{eq:first-var-seed}, the second variation formula \eqref{eq:second-variation-v}, and \eqref{eq:D-g-slash-hat} imply \eqref{eq:trivial-2}. The seventh term in \eqref{eq:bottom-sphere-estimates} is estimated directly from the Codazzi equation \eqref{eq:Codazzi1}. The eighth term is estimated directly from the Gauss equation \eqref{eq:Gauss-eqn}. The ninth term is estimated directly from the curl equation \eqref{eq:curl-eqn}. The tenth term is estimated using Gr\"onwall and the Bianchi equation \eqref{eq:Bianchi-beta-bar}. The eleventh term is estimated by using the Einstein equations \eqref{eq:alpha-traceless} and the first variation formula \eqref{eq:first-variation-v} to compute
\[0=D\tr\underline\alpha = \tr D\underline\alpha - 2\Omega(\chi,\underline\alpha).\]
Combined with the Bianchi identity \eqref{eq:Bianchi-alpha}, this yields
\[\partial_v \underline\alpha_{AB}-(\hat\chi,\underline\alpha)\slashed g{}_{AB}-\tfrac 12 \tr\chi\underline\alpha_{AB}=-(\slashed \nabla\ohat\underline\beta)_{AB}+5(\eta\ohat \underline\beta)_{AB}-3\hat{\underline\chi}_{AB}\rho + 3{}^*\hat{\underline\chi}_{AB}\sigma,\]
from which the desired estimate follows by Gr\"onwall. The twelfth term is estimated by integrating \eqref{eq:omega-bar-transport} and the thirteeth term is estimated by integrating \eqref{eq:D-bar-omega-bar-transport}. 
\end{proof}

We are now ready to prove the main result of this subsection.

\begin{proof}[Proof of \cref{prop:Schw-perturb-Schw}]
 Let $x^\mathrm{low}(v)$, $v\in [0,1]$, be the null data constructed above. We have defined
\begin{equation*}
    x^\mathrm{low}(1)=\left(1,R^2\gamma,\delta\frac{2}{R}\left(1-\frac{2M}{R}\right),0,-\frac{1}{\delta}\frac{2}{R},0,\dotsc,0\right)
\end{equation*} and we set $x^\mathrm{high}(1)\doteq 0$. 
Immediately from the definition of the boost $\mathfrak b_\delta$ in \cref{def:boost}, we have \eqref{eq:boost-2}. 

Since $\Omega^2=1$ along $C$ and $\hat\chi$ is compactly supported away from $v=0$, we have $x^\mathrm{high}(0)=0$. The boost $\mathfrak b_\delta$ changes every positive power of $\delta$ on the left-hand side of \eqref{eq:bottom-sphere-estimates} into a negative power, so that 
\begin{equation*}
    \|\mathfrak b_\delta(x(0))-\mathfrak s_{M_*,R}\|_{C^j}\les_j \delta^{1/2}
\end{equation*} 
for any $j\ge 0$. Therefore, taking $j$ sufficiently large, we have 
\begin{equation*}
   \|\mathfrak b_\delta(x(0))-\mathfrak s_{M_*,R}\|_{\mathcal X^m}\les \|\mathfrak b_\delta(x(0))-\mathfrak s_{M_*,R}\|_{C^j}\les \delta^{1/2},
\end{equation*}
where $\mathcal X^m$ is the sphere data norm appearing in \cref{thm:CR}. Now \eqref{eq:boost-est-1} follows follows by taking $\delta$ sufficiently small.
\end{proof}

\subsection{Gluing Minkowski space to any round Schwarzschild sphere}
\label{sec:thmA}

\begin{thmA}\label[thmA]{thm:Schwarzschild-gluing-2}
\hypertarget{thmA}{Let} $M>0$, $R>0$, and $k\in\Bbb N$. For any $\ve>0$ there exists a solution $x$ of the null constraints on a null cone $C^{[0,1]}$ such that $x(1)$ equals $\mathfrak s_{M,R}$ after a boost and $x(0)$ can be realized as a sphere in Minkowski space in the following sense: There exists a $C^k$ spacelike 2-sphere $S$ in Minkowski space and a choice of $L$ and $\underline L$ on $S$ such that the induced $C_u^2C^{k}_v$ sphere data on this sphere equals $x(0)$ after a boost and a sphere diffeomorphism. 
\end{thmA}

\begin{rk}
    The sphere $S$ can be made arbitrarily close to round and the sphere diffeomorphism can be made arbitrarily close to the identity. 
\end{rk}

\begin{proof}	
By scaling, it suffices to prove the theorem when $R=2$. We first use \cref{prop:Schw-perturb-Schw} to connect $\mathfrak b_\delta(\mathfrak s_{M,R})$ to the sphere data set $\mathfrak b_\delta(x(0))$ with $M_* = \ve_\sharp^{1/4}\ll 1$. We now aim to use \cref{thm:CR} to connect $x_2\doteq \mathfrak b_\delta(x(0))$ to a sphere in Minkowski space. Let $\underline x$ be the usual ingoing Minkowski null data passing through the unit sphere at $u=0$.\footnote{Note that \cref{thm:CR} was formulated for $C=[1,2]_v\times S^2$, but we are applying it on $C=[0,1]_v\times S^2$ here, which is merely a change of notation. We refer the reader back to \cref{fig:two-phases} for the setup of this proof.}

By a direct computation, $ \|\mathfrak s_{M_*,2}-\mathfrak m_2\|_{\mathcal X^m} \approx M_*.$ It follows that 
\begin{equation}
\|x_2-\mathfrak m_2\|_{\mathcal X^m} < C_1  M_*\label{eq:good-smallness-1}
\end{equation}
if $\ve_\sharp$ is sufficiently small, where $C_1$ does not depend on $\ve_\sharp$. 

We must estimate the conserved charge deviation vector
\begin{equation*}
(\Delta\mathbf E,\Delta\mathbf P,\Delta\mathbf L,\Delta\mathbf G) = (\mathbf E,\mathbf P,\mathbf L,\mathbf G)(\mathfrak b_\delta(x(0))).
\end{equation*}
By \eqref{eq:boost-est-1}, 
\begin{equation*}
|\mathbf B|+|\phi\adiv\mathbf B|\les \ve_\sharp.
\end{equation*}
We then compute 
\begin{align*}
\phi^3 \left(K+\tfrac 14 \tr\chi\tr\underline\chi\right)&= 2^3\left(\frac{1}{2^2}+\frac 14 \left(\frac{2}{2}\left(1-\frac{2M_*}{2}\right)\right)\left(-\frac{2}{2}\right)\right)+O(\ve_\sharp)\\
&= 2M_*+O(\ve_\sharp),
\end{align*}
where $O(\ve_\sharp)$ denotes a function all of whose angular derivatives are $\les \ve_\sharp$.
 It follows that for $\ve_\sharp$ sufficiently small,
\begin{equation}
\Delta\mathbf E\ge  \tfrac 32 M_* >M_* \label{eq:proof-2}
\end{equation}
and 
\begin{equation}
|\Delta\mathbf P|+|\Delta\mathbf L|+|\Delta\mathbf G|\les \ve_\sharp. \label{eq:proof-1}
\end{equation}

Let and $C_*$ and $\ve_0$ as in \cref{thm:CR} for the choice $C_\mathbf E=C_1^{-1}$. For $0<\ve_\sharp <(\ve_0/C_1)^4$, set $\ve_\flat = C_1M_*$. Then \eqref{eq:CR-cond-4} and \eqref{eq:CR-cond-1} are satisfied,
\begin{equation*}
    |\Delta \mathbf L| < C\ve_\sharp = CM_*^4 \le (C_1M_*)^2  
\end{equation*}
if $\ve_\sharp$ is sufficiently small, so \eqref{eq:CR-cond-2} is satisfied, and finally 
\begin{equation}
    |\Delta\mathbf P|+|\Delta \mathbf G| < C\ve_\sharp =   CM_*^3 \cdot M_* \le C_*\Delta\mathbf E
\end{equation}
if $\ve_\sharp$ is sufficiently small, so \eqref{eq:CR-cond-3} is also satisfied.

 By applying \cref{thm:CR}, we obtain a null data set for which the bottom sphere $x_1'$ is a sphere diffeomorphism of a genuine Minkowski sphere data set and satisfies 
 \begin{equation}
     \|x_1'-\mathfrak m_1\|_{\mathcal X} \les \|x_2-\mathfrak m_2\|_{\mathcal X}\les \ve_\sharp^{1/4},
 \end{equation}
 which can be made arbitrarily small and hence completes the proof of the theorem. 
\end{proof}

\subsection{Gluing Minkowski space to any Kerr coordinate sphere in very slowly rotating Kerr}
\label{sec:thmB}

In this section, we perform Kerr gluing for small angular momentum essentially as a corollary of the Schwarzschild work. 

\begin{thmB}\label[thm]{thm:Kerr-gluing-2} \hypertarget{thmB}{For any} $k\in\Bbb N$, there exists a function $\mathfrak a_0:(0,\infty)^2\to (0,\infty)$ with the following property.  Let $M>0$ and $R>0$. If $0\le |a| \le \mathfrak a_0(M,R)M$, there exists a solution $x$ of the null constraints on a null cone $C_0^{[0,1]}$ such that $x(1)$ equals $\mathfrak k_{M,a,R}$ after a boost and $x(0)$ can be realized as a sphere in Minkowski space in the following sense: There exists a $C^k$ spacelike 2-sphere $S$ in Minkowski space and a choice of $L$ and $\underline L$ on $S$ such that the induced $C_u^2C^{k}_v$ sphere data on this sphere equals $x(0)$ after a boost and a sphere diffeomorphism. 
\end{thmB}

\begin{proof}
Again, it suffices to prove the theorem for $R=2$ and $M>0$ fixed but otherwise arbitrary. Let $x(v)$ and $\delta$ be the associated null data set and boost parameter constructed in \cref{prop:Schw-perturb-Schw}, where $M$ and $R$ are as in the statement of the present theorem and $\ve_\sharp$ is sufficiently small that the argument of \hyperlink{thmA}{Theorem A} applies. 

Let $\tilde\gamma \doteq  2^{-2} \slashed g_{M,a,2}$ and define $\hat{\slashed g}$ by \eqref{defn:g-hat} with $\lambda=\lambda_0(\delta)$ from \cref{lem:construction-of-lambda}. By \eqref{eq:Kerr-approx}, Cauchy stability for the proof of \cref{prop:Schw-perturb-Schw}, and \cref{cor:CR}, we conclude that $\mathfrak k_{M,a,2}$ can be glued to Minkowski space as in \cref{fig:char-gluing-setup} for $a$ sufficiently small.
\end{proof}

\section{Gravitational collapse to a Kerr black hole of prescribed mass and angular momentum}\label{sec:proofs-of-corollaries}

In this section we give the proof of \cref{cor:main} and the sketch of the proof of \cref{cor:spacelike}. Recall the fractional Sobolev spaces $H^s$, $s\in\Bbb R$, and their local versions $H^s_\mathrm{loc}$. Recall also the notation $f\in H^{s-}_\mathrm{loc}$ which means $f\in H^{s'}_\mathrm{loc}$ for every $s'<s$. 

 We refer the reader back to \cref{main-corollary-diagram} for the Penrose diagram associated to the following result. 

\setcounter{cor}{0}
\begin{cor}\label{cor:main-2}
There exists a constant $\mathfrak a_0>0$ such that the following holds. For any mass $M>0$ and specific angular momentum $a$ satisfying $a/M\in[-\mathfrak a_0,\mathfrak a_0]$, there exist one-ended asymptotically flat Cauchy data $(g_0,k_0)\in H^{7/2-}\times H^{5/2-}$ for the Einstein vacuum equations \eqref{eq:EVE} on $\Sigma\cong\Bbb R^3$, satisfying the constraint equations
\begin{align}
 \label{eq:Hamiltonian}   R_{g_0} + (\tr_{g_0}k_0)^2 - |k_0|_{g_0}^2 &=0\text{ and}\\
\label{eq:momentum}   \Div_{g_0} k_0- {}^{g_0}\nabla \tr_{g_0}k_0 &=0,
\end{align} 
such that the maximal future globally hyperbolic development $(\mathcal M^4,g)$ has the following properties: 
 \begin{itemize}
     \item Null infinity $\mathcal I^+$ is complete.
     \item The black hole region is non-empty, $\mathcal B\mathcal H\doteq \mathcal M\setminus J^-(\mathcal I^+)\neq \emptyset$.
     \item The Cauchy surface $\Sigma$ lies in the causal past of future null infinity, $\Sigma\subset J^-(\mathcal I^+)$. In particular, $\Sigma$ does not intersect the event horizon $\mathcal H^+\doteq\partial(\mathcal{BH})$ or contain trapped surfaces. 
      \item $(\mathcal M,g)$ contains trapped surfaces. 
     \item For sufficiently late advanced times $v\ge v_0$, the domain of outer communication, including the event horizon, is isometric to that of a Kerr solution with parameters $M$ and $a$. For $v\ge v_0$, the event horizon of the spacetime can be identified with the event horizon of Kerr. 
 \end{itemize}
\end{cor}

\begin{rk}
    The spacetime metric $g$ is in fact $C^2$ everywhere away from the region labeled ``Cauchy stablity'' in \cref{Penrose-diagram-proof} below. Near the set $\dot{\mathcal H}^+_-$ (see \cite[p.~187]{HE73} for notation), the spacetime metric might fail to be $C^2$, but is consistent with the regularity of solutions constructed in \cite{HKM} with $s=\tfrac 72-$. See also \cite{C-low-reg} for the notion of the maximal globally hyperbolic development in low regularity. 
\end{rk}

\begin{figure}[ht]
\centering{
\def\svgwidth{25pc}
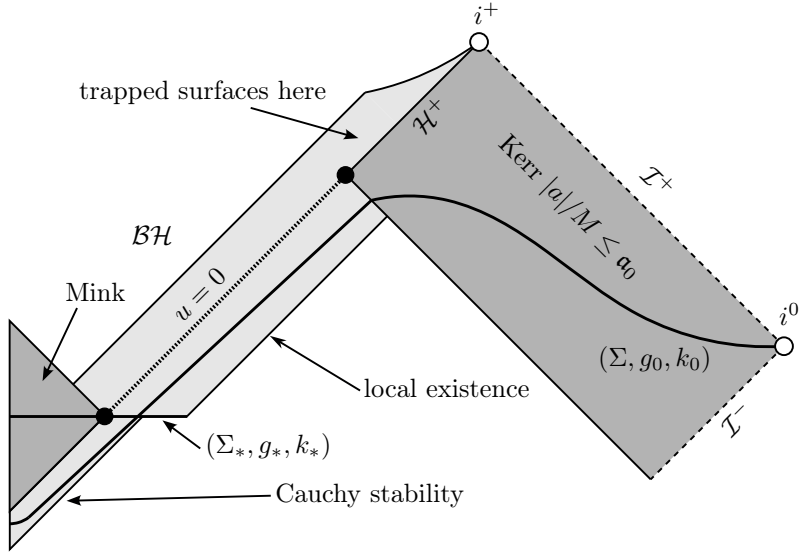}
\caption{Penrose diagram for the proof of \cref{cor:main}. The diagram does not faithfully represent the geometry of the spacetime near the ``bottom'' of the event horizon $\dot{\mathcal H}^+_-$ (i.e., the locus where the null geodesic generators ``end''). The event horizon does not necessarily end in a point since the distinguished Minkowski sphere is not necessarily round.} 
\label{Penrose-diagram-proof}
\end{figure}

\begin{proof}
We refer the reader to \cref{Penrose-diagram-proof} for the anatomy of the proof, which is essentially the same as Corollary~1 in \cite{KU22}.  The region to the left of $\mathcal H^+$ is constructed using our gluing theorem, \cref{thm:Kerr-gluing-2}, and local existence (in this case we appeal to \cite{Luk-local-existence}). See \cite[Proposition 3.1]{KU22}. The region to the right of the horizon, save for the part labeled ``Cauchy stability'' in \cref{Penrose-diagram-proof}, is constructed in the same manner. These two regions can now be pasted along $u=0$ and the resulting spacetime will be $C^2$. 

We can now use a Cauchy stability argument to construct the remainder of the spacetime. A very similar argument in carried out in \cite[Lemma 5.1]{KU22}, but the lower regularity of our gluing result in the present paper forces us to use slightly more technology here. As in \cite[Lemma 5.1]{KU22}, we take the induced data $(g_*,k_*)$ on a suitably chosen spacelike hypersurface $\Sigma_*$ passing through the bottom gluing sphere. See \cref{Penrose-diagram-proof}. This data lies in the regularity class $H^{7/2-}_\mathrm{loc}\times H^{5/2-}_\mathrm{loc}$ by part (i) of \cref{lem:Sob} below and satisfies the constraint equations. 
The cutoff argument presented in \cite[Lemma 5.1]{KU22} goes through using (ii) of \cref{lem:Sob} and the low regularity well-posedness theory in \cite{HKM}. Note that for simplicity we have applied well-posedness in the class $H^{3-}_\mathrm{loc}\times H^{2-}_\mathrm{loc}$ because of a loss of half a derivative in our Hardy inequality argument below, but since $(g_*,k_*)$ actually lies in the better space $H^{7/2-}_\mathrm{loc}\times H^{5/2-}_\mathrm{loc}$, the spacetime metric has regularity consistent with $H^{7/2-}_\mathrm{loc}\times H^{5/2-}_\mathrm{loc}$ initial data by propagation of regularity. 

To show that $(\mathcal M,g)$ contains trapped surfaces, it suffices to observe that 
$\underline D(\Omega\tr\chi)<0$ on $\mathfrak k_{M_f,a_f,r_+}$ for $\mathfrak a_0$ sufficiently small by \eqref{eq:Kerr-approx} and \eqref{eq:wave-eqn-r}.\footnote{For convenience, we have deduced the presence of trapped surfaces in very slowly rotating Kerr perturbatively from Schwarzschild. However, it is well known that any subextremal Kerr black hole contains trapped surfaces right behind the event horizon, and one may invoke that fact instead since the spacetime metric constructed here is $C^2$ across the event horizon $\mathcal H^+$.}

Having constructed the spacetime, we can finally extract a Cauchy hypersurface $\Sigma$, which completes the proof.   \end{proof}

\begin{lem}\label{lem:Sob}
 Let $f$ and $g$ be functions defined on $B_2\subset \Bbb R^3$, the ball of radius two, such that $f\in C^2(B_2)$, $f|_{B_1}\in C^3(\overline{B_1})$, $f|_{B_2\setminus\overline{B_1}}\in C^3(\overline{B_2}\setminus B_1)$, $g\in C^1(B_2)$, $g|_{B_1}\in C^2(\overline{B_1})$, and $g|_{B_2\setminus\overline{B_1}}\in C^2(\overline{B_2}\setminus B_1)$. Then:
 \begin{enumerate}[(i)]
     \item  $(f,g)\in H^{7/2-}\times H^{5/2-}(B_2)$.
     \item Suppose that $f=g=0$ identically on $B_1$. For $0<\ve<\tfrac 12$, let $\theta_\ve$ be a cutoff function which is equal to one on $B_{1+\ve}$ and zero outside of $B_{2+\ve}$. Then $f_\ve\doteq \theta_\ve f$ and $g_\ve\doteq \theta_\ve g$ satisfy 
   \begin{equation}
       \lim_{\ve\to 0}\|(f_\ve,g_\ve)\|_{H^s\times H^{s-1}(B_2)}=0\label{eq:Sob-1}
   \end{equation}
   for any $s<3$. 
 \end{enumerate}
\end{lem}

\begin{proof}
The proof of (i) follows in a straightforward manner from the physical space characterization of fractional Sobolev spaces (such as in \cite{hitchhiker}) and is effectively an elaboration of the fact that the characteristic function of $B_1$ lies in $H^{1/2-}$. 

Proof of (ii): Using Taylor's theorem as in \cite[Lemma 5.1]{KU22}, we see that $\|(f_\ve,g_\ve)\|_{H^2\times H^1(B_2)}\to 0$ as $\ve\to 0$. By iterating Hardy's inequality, we see that 
    \begin{equation*}
        \int_{B_2}f^2 |\partial^3\theta_\ve|^2\,dx\les \int_{B_{1+2\ve}\setminus B_{1+\ve}} \frac{f^2}{\ve^6}\,dx \les
        \int_{B_2} f^2+|\partial f|^2 +|\partial^2 f|^2+|\partial^3 f|^2\,dx,
    \end{equation*}
    so $(f_\ve,g_\ve)$ is bounded in $H^3\times H^2$. We now obtain \eqref{eq:Sob-1} by interpolation. 
\end{proof}
\begin{rk}
    In fact \eqref{eq:Sob-1} holds for $s<\frac 72$, but this requires a fractional Hardy inequality. 
\end{rk}

We now sketch the proof of \cref{cor:spacelike} and refer the reader back to \cref{fig:spacelike-singularity} for the associated Penrose diagram.

\begin{proof}[Sketch of the proof of \cref{cor:spacelike}]
Using \hyperlink{thmA}{Theorem A}, we glue Minkowski space to a round Schwarzschild sphere of mass $1$ and radius $R=2-\ve$ for $0\le\ve\ll 1$. As $\ve\to 0$ (perhaps only along a subsequence $\ve_j\to 0$), the gluing data converge to the horizon gluing data used in the proof of \cref{cor:main-2}, in an appropriate norm. It then follows by Cauchy stability that the spacetimes constructed by solving backwards as in the proof of \cref{cor:main-2} contain the full event horizon, for $\ve$ sufficiently small. 
\end{proof}

\begin{appendix}
\section{Double null gauge} \label{app:A}

In this appendix we briefly recall the basic notion of \emph{double null gauge} \cite{double-null, Christo09}.

\subsection{Spacetimes in double null gauge}

\subsubsection{Double null gauge} 

Let $\mathcal W\subset \Bbb R^2_{u,v}$ be a domain and define $\mathcal M^{3+1}\doteq \mathcal W\times S^2$. Denote $S_{u,v}\doteq\{(u,v)\}\times S^2\subset \mathcal M$. The distinguished foliation of $\mathcal M$ by these spheres carries a tangent bundle $TS $ and cotangent bundle $T^*S\doteq(TS)^*$. An \emph{$S$-tensor (field)} is a section of a vector bundle consisting of tensor products of $TS$ and $T^*S$. Let $\slashed g$ be a positive-definite $(0,2)$ $S$-tensor field, let $\Omega^2$ be a positive function on $\mathcal M$, and $b$ be an $S$-vector field. Under these assumptions, the formula
\begin{equation*}
    g= -4\Omega^2\,du\,dv + \slashed g{}_{AB}(d\vartheta^A-b^A\,dv)(d\vartheta^B-b^B\,dv)
\end{equation*}
defines a Lorentzian metric on $\mathcal M$, where $(\vartheta^1,\vartheta^2)$ are arbitrary local coordinates on $S^2$ and $\slashed g{}_{AB}$ (resp., $b^A$) are the components of $\slashed g$ (resp., $b$) relative to this coordinate basis. The coordinate functions $u$ and $v$ satisfy the eikonal equation, i.e.,
\begin{equation}
    g^{\mu\nu}\partial_\mu u\partial_\nu u=0\quad\text{and}\quad g^{\mu\nu}\partial_\mu v \partial_\nu v=0.\nonumber
\end{equation}
Consequently, the hypersurfaces $C_u\doteq\{u=\mathrm{const.}\}$ and $\underline C_v\doteq \{v=\mathrm{const.}\}$ are null hypersurfaces. We time orient $(\mathcal M,g)$ by declaring $\partial_u+\partial_v+b^A\partial_{\vartheta^A}$ to be future-directed. 

The vector fields
\begin{align*}
    L'\doteq -2(du)^\sharp \quad\text{and}\quad
    \underline L'\doteq -2(dv)^\sharp
\end{align*}
are future-directed null geodesic vector fields. We set 
\[L\doteq\Omega^2 L'\quad\text{and}\quad \underline L \doteq \Omega^2\underline L',\]
which then satisfy 
\begin{align*}
   Lu &=  0,\quad Lv =1,\\
   \underline Lu &=1 ,\quad \underline Lv=0.
\end{align*}
Finally, we set 
\[e_4\doteq \Omega L',\quad e_3\doteq \Omega\underline L'.\]
Given arbitrary coordinates $\vartheta^A$ on $S^2$ and defining $e_A\doteq \partial_{\vartheta^A}$, the quadruple $\{e_1,e_2,e_3,e_4\}$ is called a (normalized) \emph{null frame}, which satisfies 
\begin{align}
    g(e_A,e_3)=g(e_A,e_4)=0,\quad  g(e_3,e_4)=-2, \quad g(e_3,e_3)=g(e_4,e_4)=0.\nonumber
\end{align}

\subsubsection{Algebra and calculus of \texorpdfstring{$S$}{S}-tensors}

Let $(\mathcal M,g)$ be a spacetime equipped with a double null gauge as above. For vector fields on $\mathcal M$, we define the orthogonal projection to $S$ vector fields by 
\[\Pi :T\mathcal M\to TS,\quad \Pi X \doteq X+\tfrac 12g(X,e_3)e_4+\tfrac 12g(X,e_4)e_3\]
which we extend componentwise to contravariant tensors of higher rank. We note that $   \Pi \circ i= \mathrm{id}$ on $TS$, where $i\colon TS \subset T\mathcal M$ is the natural inclusion. By duality, this defines a ``promotion'' operator $\Pi^\ast \colon T^*S \to T^*\mathcal M$ which extends componentwise to covariant $S$-tensors and satisfies $i^*\circ \Pi^* = \mathrm{id}$ on $T^\ast S$.

We now define projected Lie derivatives $\slashed{\mathcal L}_L$ and $\slashed{\mathcal L}_{\underline L}$ on $S$-tensors. If $X$ is an $S$-vector field, then 
\[\slashed{\mathcal L}_L X\doteq \mathcal L_L X,\quad \slashed{\mathcal L}_{\underline L}X\doteq \mathcal L_{\underline L}X\]
are already $S$-vector fields. If $\xi$ is an $S$-1-form, then 
\[\slashed{\mathcal L}_L\xi \doteq i^\ast  \mathcal L_L ( \Pi^* \xi) =  \mathcal L_L(\xi\circ \Pi)|_{TS},\quad\slashed{\mathcal L}_{\underline L}\xi \doteq i^\ast  \mathcal L_{\underline L} ( \Pi^* \xi) = \mathcal L_{\underline L}(\xi\circ \Pi)|_{TS}, \]
where we have explicitly written the ``promotion'' operation which will be consistently omitted in the sequel. The operation is extended to general $S$-tensor fields via the Leibniz rule. As a shorthand, we write 
\[D\doteq \slashed{\mathcal L}_L,\quad \underline D\doteq \slashed{\mathcal L}_{\underline L}.\]

The symbol $\slashed\nabla$ acts on functions and $S$-vector fields as the induced covariant derivative on the spheres and is extended to general $S$-tensors by the Leibniz rule.

We will frequently make use of the following notation: 
Let $\xi,\eta$ be $S$-1-forms and $\theta,\phi$ symmetric covariant $S$-2-tensor fields. We then define
\begin{align*}
    (\xi\ohat \eta)_{AB}&\doteq \xi_A\eta_B + \xi_B \eta_A - (\xi\cdot\eta)\slashed g{}_{AB} \\
    (\slashed\nabla\ohat \xi)_{AB}&\doteq \slashed\nabla_A\xi_B+\slashed\nabla_B\xi_A - (\adiv\xi)\slashed g{}_{AB}\\
   \adiv \xi &\doteq \slashed g{}^{AB}\slashed\nabla_A\xi_B\\
   \arot\xi &\doteq \slashed\varepsilon{}^{AB}\slashed\nabla_A \xi_B\\
   ({}^*\xi)_A&\doteq \slashed\ve{}_{AB}\slashed g{}^{BC}\xi_C\\
  \hat\theta_{AB} &\doteq \theta_{AB}- \tfrac 12 \tr\theta \,\slashed g{}_{AB}\\ 
  \theta\wedge\phi&\doteq \slashed\ve{}^{AB}\slashed g{}^{CD} \theta_{AC}\phi_{BD},
\end{align*}
where $\slashed\ve$ is the induced volume form on $S_{u,v}$. 
The notation $\slashed g{}^{AB}$ denotes the inverse of the induced metric $\slashed g{}_{AB}$. Indices of $S$-tensors are raised and lowered with $\slashed g$ and $\slashed g{}^{-1}$. 

\subsubsection{Ricci and curvature components}

The Ricci components are given by the null second fundamental forms 
\begin{equation*}
    \chi_{AB} \doteq g(\nabla_A e_4,e_B),\quad  \underline\chi_{AB} \doteq g(\nabla_A e_3, e_B),
\end{equation*}
the torsions
\begin{equation*}
    \eta_A\doteq -\tfrac 12 g(\nabla_3 e_A,e_4), \quad \underline\eta_A\doteq -\tfrac 12 g(\nabla_4 e_A,e_3),   
\end{equation*}
and 
\begin{equation*}
    \omega \doteq D\log\Omega, \quad \underline{\omega} \doteq \underline D\log\Omega.
\end{equation*}
The 1-form 
\[\zeta \doteq \eta - \slashed\nabla \log\Omega\]
is also referred to as the torsion. Note that 
\begin{equation}
    \eta+\underline \eta = 2\slashed\nabla\log\Omega.\label{eq:eta-etabar} 
\end{equation}

The null curvature components are given by 
\begin{align*}
\alpha_{AB}    &\doteq R(e_A,e_4,e_B,e_4), & \underline\alpha_{AB} &\doteq R(e_A,e_3,e_B,e_3), \\
\beta_A    &\doteq \tfrac 12 R(e_A,e_4,e_3,e_4), & \underline\beta_A&\doteq \tfrac 12 R(e_A,e_3,e_3,e_4),\\
\rho     &\doteq \tfrac 14 R(e_4,e_3,e_4,e_3), & \sigma&\doteq \tfrac 14 {}^*R(e_4,e_3,e_4,e_3),
\end{align*}
where $R(W,Z,X,Y)=g(R(X,Y)Z,W)$ is the Riemann tensor.

\subsubsection{Normalized sphere data determined by a geometric sphere}

For the notion of sphere data used here, see \cref{sec:sphere-data}. 

\begin{lem}\label{lem:data-generation}
    Let $(\mathcal M^4,g)$ be a spacetime satisfying the Einstein vacuum equations \eqref{eq:EVE} and $i:S^2 \to \mathcal M$ an embedding with spacelike image $S\doteq i(S^2)$. Let $\underline L$ be a null vector field along $S$ which is normal to $S$. Then for any $m\ge 0$ there exists a unique associated $C^2_uC^{2+m}_v$ sphere data set $x[g,i,\underline L]$ such that $\slashed g=i^*g$, $\Omega^2=1$, and  $\omega=D\omega=\cdots = D^{m+1}\omega=\underline\omega=\underline D\underline\omega=0$.

    The sphere data $x[g,i,\underline L]$ depends smoothly on $(g,i,\underline L)$ in the natural way. 
\end{lem}
We say that $x[g,i,\underline L]$ is \emph{generated} by $(g,i,\underline L)$. If $\psi:S^2\to S^2$ is a diffeomorphism, then $x[g,i\circ\psi,\underline L]$ is related to $x[g,i,\underline L]$ by a sphere diffeomorphism as in \cref{def:sphere-diffeo}. 
\begin{proof}
    The geometric sphere $S$ is identified with the round sphere by $i$, which endows $S$ with a choice of round metric $\gamma$. The dual null vector field $L$ is uniquely determined by the requirement that $L\perp TS$ and $g(L,\underline L)=-2$. Let $C\cup \underline C$ be the (locally defined) bifurcate null hypersurface passing through $S$ such that $L$ is tangent to $C$ and $\underline C$ is tangent to $\underline L$. Let $\Omega^2=1$ identically on $C\cup\underline C$. Given this data, there is a unique double null foliation with respect to $g$ covering a neighborhood of $S$ in $\mathcal M$. Now $x[g,i,\underline L]$ is constructed by computing the corresponding quantities in this double null foliation and taking the values at $S$. 
\end{proof}

\subsection{The Einstein equations in double null gauge}

We now assume that the spacetime metric $g$ satisfies the Einstein vacuum equations \eqref{eq:EVE}. In double null gauge, the Einstein equations are equivalent to the \emph{null structure equations} (with the Ricci coefficients on the left-hand side) and the \emph{Bianchi equations} (with the curvature components on the left-hand side). The Einstein equations \eqref{eq:EVE} imply 
\begin{equation}
    \tr\alpha=0,\qquad \tr\underline\alpha=0. \label{eq:alpha-traceless}
\end{equation}

\subsubsection{The null structure equations}

First variation formulas:
\begin{align}
 D\slashed g &= 2\Omega\chi = 2\Omega\hat\chi + \Omega \tr\chi\,\slashed g\label{eq:first-variation-v}\\
\underline D\slashed g &= 2\Omega\hat{\underline\chi} = 2\Omega\hat{\underline\chi}+2\Omega\tr\underline\chi\,\slashed g
\end{align}
Raychaudhuri's equations: 
\begin{align}
D\tr \chi +\tfrac 12 \Omega (\tr\chi)^2 -\omega\tr\chi &= - \Omega|\hat\chi|^2 \label{eq:Ray-v}\\
 \underline D\tr\underline\chi +\tfrac 12 \Omega (\tr\underline\chi)^2 -\underline\omega\tr\underline \chi &=-\Omega |\hat{\underline\chi}|^2\label{eq:Ray-u}
\end{align}
Transport equations for Ricci components:
\begin{align}
\label{eq:second-variation-v} D\hat\chi &= \Omega|\hat\chi|^2 \slashed g +\omega\hat\chi -\Omega\alpha\\
\label{eq:second-variation-u}\underline D\hat{\underline\chi} &= \Omega|\hat{\underline\chi}|^2 \slashed g +\underline\omega\hat{\underline\chi} -\Omega\underline\alpha\\
 D\eta &= \Omega(\chi\cdot \underline \eta-\beta)\label{eq:eta-transport}\\
 \underline D \underline \eta &= \Omega(\underline\chi \cdot\eta +\underline\beta)\\
 \label{eq:omega-bar-transport}D\underline\omega &= \Omega^2 (2\eta\cdot\underline\eta-|\eta|^2-\rho)\\
 \underline D\omega &= \Omega^2 (2\eta\cdot\underline\eta-|\eta|^2-\rho)\\
 D\underline \eta &= -\Omega(\chi\cdot\underline\eta -\beta) + 2\slashed\nabla \omega\\
\underline D \eta &= -\Omega(\underline\chi\cdot\eta -\beta) + 2\slashed\nabla \omega\\
\label{eq:tr-chi-bar-transport} D(\Omega\tr\underline \chi)&=2\Omega^2\adiv\underline\eta+2\Omega^2|\underline\eta|^2-\Omega^2 (\hat\chi,\hat{\underline\chi})-\tfrac 12 \Omega^2 \tr\chi\,\tr\underline\chi + 2\Omega^2\rho\\
\label{eq:wave-eqn-r}\underline D(\Omega\tr\chi)&=2\Omega^2\adiv\eta+2\Omega^2|\eta|^2-\Omega^2 (\hat\chi,\hat{\underline\chi})-\tfrac 12 \Omega^2 \tr\chi\,\tr\underline\chi + 2\Omega^2\rho\\
\label{eq:chi-bar-hat-transport}D(\Omega\hat{\underline\chi})&=\Omega^2 \left((\hat\chi,\hat{\underline\chi})+\tfrac 12 \tr\chi\,\hat{\underline\chi}+\slashed \nabla \ohat \underline\eta + \underline\eta\ohat\underline\eta - \tfrac 12 \tr\underline\chi \,\hat\chi \right)\\
\underline D(\Omega\hat{\chi})&=\Omega^2 \left((\hat\chi,\hat{\underline\chi})+\tfrac 12 \tr\underline\chi\,\hat{\chi}+\slashed \nabla \ohat \eta + \eta\ohat\eta - \tfrac 12 \tr\chi \,\hat{\underline\chi} \right)
\end{align}
Gauss equation:
\begin{equation}
    K+\tfrac 14 \tr\chi\,\tr\underline\chi - \tfrac 12 (\hat\chi,\hat{\underline\chi}) = -\rho\label{eq:Gauss-eqn}
\end{equation}
Codazzi equations:
\begin{align}
 \label{eq:Codazzi1}  \adiv \hat\chi - \tfrac 12 \slashed\nabla \tr\chi + \hat \chi\cdot \zeta - \tfrac 12 \tr\chi\,\zeta &=-\beta\\
\label{eq:Codazzi2}  \adiv \hat{\underline\chi} - \tfrac 12 \slashed\nabla \tr\underline\chi - \hat{\underline\chi} \cdot\zeta + \tfrac 12 \tr\underline\chi\, \zeta &= \underline\beta 
\end{align}
Curl equations:
\begin{equation}
    \arot\eta = -\arot\underline\eta = \arot\zeta = - \tfrac 12 \hat\chi\wedge\hat{\underline\chi} -\sigma\label{eq:curl-eqn}
\end{equation}
We also require
\begin{align}
\label{eq:D-bar-omega-bar-transport}D\underline D\underline\omega =&-12\Omega^2(\eta-\slashed\nabla \log\Omega, \slashed\nabla\underline\omega) +2\Omega^2\underline\omega \left((\eta,-3\eta+4\slashed\nabla\log\Omega)-\rho\right)\\
& + 4\Omega^3 \underline\chi (\eta,\slashed\nabla\log\Omega)+\Omega^3(\underline\beta, 7\eta - 3\slashed\nabla\log\Omega)+\tfrac 32 \Omega^3\rho \tr\underline\chi + \Omega^3\adiv\underline\beta + \tfrac 12 \Omega^3 (\hat\chi,\underline\alpha). \nonumber
\end{align}

\subsubsection{The Bianchi identities}

In this paper, we only need the following two Bianchi identities: 
\begin{align}
  \label{eq:Bianchi-alpha}  \hat{ D}\underline\alpha +\left(2\omega- \tfrac 12 \Omega\tr\chi\right)\underline\alpha&=\Omega\left(-\slashed\nabla\ohat \underline\beta -(4\underline\eta-\zeta)\ohat\zeta-3\hat{\underline\chi}\rho+3{}^*\hat{\underline\chi}\sigma\right),\\
    D\underline\beta +\left(\tfrac 12 \Omega\tr\chi -\Omega\hat\chi +\omega\right)\underline\beta &= \Omega\left(-\slashed\nabla\rho +{}^*\slashed\nabla\sigma - 3\underline\eta\rho + 3{}^*\underline\eta\sigma +2\hat{\underline\chi}\cdot\beta\right).\label{eq:Bianchi-beta-bar}
\end{align}

\end{appendix}

\printbibliography[heading=bibintoc]

\end{document}